\documentclass[11pt]{article}

\usepackage{algorithm}
\usepackage{algpseudocode}

\usepackage{algorithmicx}
\algdef{SE}[DOWHILE]{Do}{doWhile}{\algorithmicdo}[1]{\algorithmicwhile\ #1}%

\algrenewcommand\algorithmicrequire{\textbf{Input:}}
\algrenewcommand\algorithmicensure{\textbf{Output:}}

\usepackage{amsmath, amsthm, amssymb, amsfonts}
\usepackage{array}
\usepackage{arydshln}

\usepackage{booktabs}

\usepackage{comment}

\usepackage{caption}
\usepackage{xurl}
\usepackage{enumitem}

\usepackage{float}

\usepackage[margin=1in,letterpaper]{geometry}
\usepackage{graphicx}

\usepackage{hyperref}
\hypersetup{colorlinks,urlcolor=blue,citecolor=black}

\usepackage{mathtools, mathrsfs}
\usepackage{multirow}

\usepackage[sort, authoryear]{natbib}

\usepackage[section]{placeins}

\usepackage{sectsty}
\sectionfont{\Large}
\subsectionfont{\large}

\usepackage{setspace}
\usepackage{subfigure}

\usepackage{textcomp}
\usepackage{titling}
\usepackage{times}

\usepackage{xcolor}

\DeclareMathOperator*{\argmin}{arg\,min}

\DeclareMathOperator*{\Minimize}{Minimize}

\newtheorem{thm}{Theorem}[section]
\newtheorem{prop}[thm]{Proposition}
\newtheorem{lem}{Lemma}[section]
\newtheorem{remark}{Remark}

\numberwithin{equation}{section}
\graphicspath{{figures/}} 
\providecommand{\keywords}[1]{\textbf{Keywords:} #1}

\title{Low-Rank Regularization of Global Fr\'{e}chet Regression Models for Distributional Responses}

\author{
    Kyunghee Han\thanks{Co-first authors} \\University of Illinois Chicago 
    \and
    Hsin-Hsiung Huang\footnotemark[1] 
    \thanks{\color{black} We gratefully acknowledge the support of the National Science Foundation through grants (DMS-1924792 and DMS-2318925).} \\University of Central Florida
}

\date{\today}

\begin{document}

\maketitle

\begin{abstract}
Fr\'echet regression has emerged as a useful tool for modeling non-Euclidean response variables associated with Euclidean covariates. In this work, we propose a global Fr\'echet regression estimation method that incorporates low-rank regularization. Focusing on distribution function responses, we demonstrate that leveraging the low-rank structure of the model parameters enhances both the efficiency and accuracy of model fitting. Through theoretical analysis of the large-sample properties, we show that the proposed method enables more robust modeling and estimation than standard dimension reduction techniques. To support our findings, we also present numerical experiments that evaluate the finite-sample performance.
\end{abstract}

\noindent\keywords{Fr\'{e}chet regression, Wasserstein space, distribution function responses, low-rank regularization, manifold learning}

\section{Introduction} \label{sec:introduction}

Advances in modern data collection have led to the emergence of complex, non-Euclidean data structures across a variety of scientific domains. Examples include graph Laplacians from network analysis \citep{chung1997spectral, kolaczyk2009statistical}, covariance matrices representing brain connectivity \citep{ferreira2013}, and distributions of voxel intensities in medical imaging such as CT hematoma studies \citep{muschelli2015ct}. These data types inhabit general metric spaces that lack the algebraic properties of Euclidean spaces, such as vector addition and inner products, making them incompatible with traditional regression techniques. Classical methods rely on linear operations and are ill-equipped to model the geometric complexity of such responses. Fr\'echet regression was developed to address this challenge by modeling the conditional Fr\'echet mean of a non-Euclidean response given Euclidean predictors \citep{petersen2019frechet}, offering a natural extension of regression to metric spaces.

Despite its theoretical appeal, Fr\'echet regression faces practical difficulties in high-dimensional settings. In such cases, standard estimation procedures become unstable, as inflated variance and near-singularity of empirical covariance matrices introduce ill-posedness and poor generalization. Various strategies have been proposed to mitigate these issues. These include dimension reduction techniques \citep{zhang2023dimension, huang2024fr}, ridge-type regularization \citep{tucker2023variable}, nonlinear Bayesian regularization of functional models \citep{yang2020quantile}, shape constraints \citep{ghosal2023shape, ghosal2025distributional}, and total variation penalties \citep{lin2021total}. While these approaches stabilize estimation by regularizing the design space or controlling complexity, they primarily target the conditional mean response and do not explicitly address the model's structural complexity or leverage low-dimensional manifolds in the response space. In the standard multivariate regression, low-rank structure in the coefficient matrix has been shown to improve estimation and interpretability by effectively reducing dimensionality \citep{she2017robust}.

The literature on low-rank regression has largely focused on Euclidean and vector-valued outcomes. For instance, She and Chen \citep{she2017robust} introduced a robust estimation framework using a sparse mean-shift parameterization. Rabusseau and Kadri \citep{rabusseau2016low} studied multilinear rank constraints in tensor regression. Kong et al. \citep{kong2020l2rm} proposed a low-rank linear model to connect high-dimensional matrices and covariates. Liu et al. \citep{liu2022multivariate} developed nested reduced-rank models for multivariate functional responses, while Huang et al. \citep{huang2024framework} advanced regularized matrix regression for classification and prediction with matrix-valued covariates. Lin et al. \citep{lin2019reduced} further extended low-rank models to functional regression within reproducing kernel Hilbert spaces. However, low-rank modeling for metric-space-valued responses remains relatively unexplored. Recent developments include Zhang et al. \citep{zhang2023dimension}, who proposed sufficient dimension reduction for Fr\'echet regression to address the curse of dimensionality and enhance visualization. Song and Han \citep{song2023errors} introduced low-rank approximations for Fr\'echet regression under measurement error, reducing the bias introduced by noisy covariates. These efforts mark progress toward integrating low-rank structures with Fr\'echet regression, yet a unified framework for distribution-valued responses remains lacking.

In this work, we introduce a novel framework for low-rank Fr\'echet regression models with distribution function responses. Our methodology bridges functional data analysis \citep{ramsay2005functional} and optimal transport theory \citep{villani2009optimal} by formulating the low-rank regularization of Fr\'echet regression in Wasserstein space. This ensures that the low-rank approximation respects the underlying geometric structure of the space of distributions. Unlike mean-based Fr\'echet regression, our quantile-based approach captures the full distributional response and provides a more nuanced understanding of covariate effects, particularly in heterogeneous settings \citep{yang2020quantile}. To improve interpretability and model flexibility, we impose $\ell_1$ and fused lasso penalties, encouraging sparsity and piecewise smoothness in the estimated coefficient functions. These penalties yield structured and interpretable coefficient surfaces that adapt to localized variations. For optimization, we adapt Riemannian algorithms from recent work on low-rank matrix regression \citep{huang2024framework} and Fr\'echet dimension reduction \citep{huang2024fr}, ensuring efficient computation even under non-Euclidean constraints. Our framework unifies quantile-based modeling with manifold-aware dimension reduction and provides theoretical guarantees of consistency and convergence. This makes it a powerful tool for analyzing complex distribution-valued data across diverse applications.

The remainder of the article is organized as follows. Section \ref{sec:methodology} introduces the proposed methodology and main results. Estimation procedures and theoretical properties are detailed in Sections \ref{sec:model-estimation} and \ref{sec:theory}. Section \ref{sec:algorithm} presents the computational algorithm for low-rank Fr\'echet quantile regression with regularization. Simulation results are reported in Section \ref{sec:simulation}, and two data applications are illustrated in Section~\ref{sec:data}. Technical details and proofs are provided in the Appendix.

\section{Methodology} \label{sec:methodology}

\subsection{Model and Estimation} \label{sec:model-estimation}
Let $\mathcal{F}$ be a family of probability distribution functions defined on a common support $\mathcal{S} \subset \mathbb{R}$.
Suppose $Y$ is a random element in $\mathcal{F}$ coupled with a random vector $X$ having its values on a compact set $\mathcal{X} \subset \mathbb{R}^p$.
Without loss of generality, we assume that the distribution of $X$ is defined on $[0, 1]^p$.
We consider the global Fr\'echet regression model \citep{petersen2019frechet} to associate $Y$ with $X$ as
\begin{align} \label{model:global-frechet}
    F(\cdot | x) = \argmin_{y \in \mathcal{F}} \mathit{E} \big[ s(X, x) \, d_{W_2}(Y, y)^2 \big] \quad (x \in [0, 1]^p),
\end{align}
where $s(X, x) = 1 + (X - \mu)^\top \Sigma^{-1} (x - \mu)$ is a weight function defined with $\mu = \mathit{E} X$ and $\Sigma = \mathrm{Var}(X)$.
Also, $d_{W_2}(y_1, y_2)^2 = \int_0^1 \big(Q_{y_1}(u) - Q_{y_2}(u) \big)^2 \, \mathrm{d}u$ denotes the square of the $2$-Wasserstein distance between $y_1, y_2 \in \mathcal{F}$, where $Q_y$ denotes the quantile function of the distribution function $y \in \mathcal{F}$, i.e., $Q_y(u) = \inf \{s \in \mathcal{S}: y(s) \geq u \}$, or equivalently $Q_y(u) = y^{-1}(u)$,  for $u \in [0,1]$.
Song and Han \citep{song2023errors} proposed a low-rank approximation of \eqref{model:global-frechet} where the weight function is not well-defined with a singular $\Sigma$.
In practice, one can also encounter similar situations when the number of covariates is greater than the sample size in a finite-sample study, but such high-dimensional problems were not covered by \citep{song2023errors}.

In this project, we will fill this gap based on the regularization framework proposed by Huang et al. \citep{huang2024framework}.
To this end, we note that the global Fr\'echet model $F(\cdot | x)$ in \eqref{model:global-frechet} is the inverse image of the function $u \mapsto Q(u | x) = \mathit{E}\big(Q_Y(u) \,|\, X = x\big)$ under some map $\psi: \mathcal{F} \to \mathcal{Q}$ defined as $\psi(y) = Q_y$.
Petersen and M\"uller \citep{petersen2016functional} explored a similar approach for functional data analysis of distribution functions via transformation to a Hilbert space.
This enables us to identify the latent model parameters through $\psi$ as follows:
\begin{align} \label{model:quantile-reparam}
    Q(u | x) = \alpha(u) + \beta(u)^\top (x - \mu). 
\end{align}
We note that the model \eqref{model:quantile-reparam} differs from the standard functional linear regression models because the outcome $Q_Y(u)$ must be monotone increasing with respect to $u$, and it is not trivial to directly identify the functional parameters $\alpha(u)$ and $\beta(u)$ that yield monotone outcomes in the presence of the stochastic component $\varepsilon(u)$.
Moreover, the collection $\mathcal{Q} = \{ Q_y: y \in \mathcal{F} \}$ of quantile functions does not form a linear space.

\begin{remark} \label{rmk:quantile-reparam}
    The reparameterization in \eqref{model:quantile-reparam} offers a more straightforward interpretation than the original formulation in \eqref{model:global-frechet}, where $\beta_j(u)$ represents the expected increase in the quantile outcome $Q_Y(u)$ for a unit increase in the $j$-th covariate for each $u \in [0, 1]$. 
    Furthermore, from \eqref{model:global-frechet}, the marginal Fr\'echet mean of $Y$ is given by $F(\cdot|\mu)$, and it holds that $Q_{F(\cdot|\mu)}(u) = \alpha(u)$ for each $u \in [0, 1]$. This implies that $u \mapsto \alpha(u)$ corresponds to the quantile function of the marginal Fr\'echet mean of $Y$.
    This convention is also well-aligned with the standard functional linear regression models. 
\end{remark}

We consider the optimal transport approach \citep{villani2009optimal} to investigate the theoretical property of the proposed method, as explored by \citep{petersen2021wasserstein}. 
Specifically, for a given \( X = x \), let $T = \{ T(s): s \in \mathcal{S} \}$ be a stochastic process such that $Q_Y(u) = T\big( Q(u | x) \big)$ 
for $u \in [0, 1]$.
Since the $u$-th quantile of $F(\cdot | x)$ is given by $Q(u | x)$ under the latent model \eqref{model:quantile-reparam}, substituting \(s = Q(u | x)\) for each \(u \in [0, 1]\) yields
\begin{align}
    T(s) = Q_Y\big( F(s|x) \big) = \big(Y^{-1} \circ F(\cdot|x)\big)(s)
\end{align}
for every $s \in \mathcal{S}$.
This indicates that $T: \mathcal{S} \to \mathcal{S}$ corresponds to Monge's formulation of the optimal transport map pushforwarding $F(\cdot|x)$ to the distribution function response $Y$ under the squared distance loss \citep{rachev2006mass, villani2009optimal}.
While the transportation does not operate additively for the space $\mathcal{F}$ of distribution functions, we may define a pseudo-error in the space $\mathcal{Q}$ of quantile functions as 
\begin{align} \label{OT:pseudo-error}
    \varepsilon(u) = Q_Y(u) - Q(u | x) = (T - \mathit{Id}) \big( Q(u|x) \big), 
\end{align}
where $\mathit{Id}: \mathcal{S} \to \mathcal{S}$ is the identity map.
We note that introducing the pseudo-error offers an alternative pathway for analyzing the statistical properties of the global Fr\'echet regression for distribution function outcomes within the framework of standard functional linear regression, as formulated in \eqref{model:quantile-reparam}.
For example, $ \mathbb{E} \big( \varepsilon(u) | X = x \big) = 0 $ holds because substituting $ s = Q(u | x) $ with $ T(s) = \big(Y^{-1} \circ F(\cdot|x)\big)(s)$ yields $ \mathbb{E}\big( T(s) \,|\, X = x \big) = \mathrm{Id}(s) $.
Moreover, it follows from \eqref{OT:pseudo-error} that
\begin{align} \label{OT:error-cov}
    \begin{split}
        \mathrm{Cov}\big( \varepsilon(u), \varepsilon(u') \big) 
        &= \mathrm{Cov}\Big( T\big(Q(u | X)\big), T\big(Q(u' | X)\big) \Big).
    \end{split}
\end{align}
In the following section, we provide some technical conditions in detail.

Next, we introduce an estimation procedure.
For a random sample $\mathcal{X}_n = \{ (Y_i, X_i): i = 1, \ldots, n \}$ of $(Y, X)$, let $\widehat{Q}(\cdot | x)$ denote the empirical estimator of the quantile function associated with $F( \cdot | x)$  in \eqref{model:global-frechet}.
It can be verified that
\begin{align} \label{prob:ill-posed}
    \begin{split}
        \widehat{Q}(u | x)
        &= \frac{1}{n} \sum_{i = 1}^n \Big( 1 + (X_i - \bar{X})^\top \widehat\Sigma^{-1} (x - \bar{X}) \Big) Q_{Y_i}(u)\\
        &= \overline{Q}_Y(u) + \widehat\Gamma(u)^\top \widehat\Sigma^{-1}(x - \bar{X}) 
    \end{split}
\end{align} 
for $u \in [0, 1]$, where $\overline{Q}_Y(u) = n^{-1} \sum_{i = 1} Q_{Y_i}(u)$, $\bar{X} = (\bar{X}_1, \ldots, \bar{X}_p)^\top$ with $\bar{X}_j = n^{-1} \sum_{i = 1}^n X_{i, j}$, $\widehat\Sigma = n^{-1} \sum_{i = 1}^n (X_i - \bar{X}) (X_i - \bar{X})^\top$, and $\widehat\Gamma(u) = n^{-1} \sum_{i = 1}^n (X_i - \bar{X}) Q_{Y_i}(u)$. 
It means that the global Fr\'echet regression model \eqref{model:global-frechet} fits the quantile function $Q_{Y_i}(u) = Y_i^{-1}(u)$ as $\widehat{Q}(u | X_i) = \overline{Q}_Y(u) + \widehat\Gamma(u)^\top \widehat\Sigma^{-1} (X_i - \bar{X})$, which corresponds to the least-squares fit of the latent model \eqref{model:quantile-reparam}, where $\hat\alpha(u) = \overline{Q}_Y(u)$ and $\hat\beta(u) = \widehat\Sigma^{-1} \widehat\Gamma(u)$ is the least-squares estimator of $\alpha(u)$ and $\beta(u)$, respectively, obtained by regressing $Q_{Y_i}(u)$ on $(1, X_i - \bar{X})$. 
However, when observed covariates are highly correlated, the empirical estimator \eqref{prob:ill-posed} becomes ill-posed as $\widehat\Sigma$ is not invertible in a finite sample study, commonly encountered with high-dimensional covariates. 
To deal with the ill-posed problem, one may first consider the following constrained optimization problem: 
\begin{align} \label{optim:functional-penalized}
    \begin{split}
        \Minimize_{\beta} \quad 
        &\sum_{i = 1}^n \int_0^1 \Big( Q_{Y_i}(u) - \overline{Q}_Y(u) - \beta(u)^\top (X_i - \bar{X}) \Big)^2 \, \mathrm{d}u \\
        \textrm{subject to} \quad 
        &\mathrm{rank}(\beta) = r
    \end{split}
\end{align}
with a pre-specified non-negative integer $r \leq p$, where $\mathrm{rank}(\beta) = \mathrm{rank}\big( \{ \beta(u) \in \mathbb{R}^p: u \in [0,1] \} \big)$.
Let $\hat{\beta} = (\hat\beta_1, \ldots, \hat\beta_p)^\top$ denote the solution of \eqref{optim:functional-penalized}.
For numerical implementation in practice, one may consider discretizing quantile functions on an equally-spaced grid $\mathcal{U}_M = \{ u_m \in (0, 1): 0 = u_0 < \cdots < u_M = 1 \}$ with a sufficiently large $M \geq 1$.

In this study, we consider discretizing the objective function in \eqref{optim:functional-penalized} and further introduce penalization to deal with high-dimensional problems. 
To this end, let $\mathbf{B} \in \mathbb{R}^{p \times M}$ be a $p \times M$ matrix whose $m$-th column represents $\beta(u_m) \in \mathbb{R}^p$ for $m = 1, \ldots, M$. 
Letting $b_{j, m}$ denote the $(j,m)$-element of $\mathbf{B}$, we instead propose solving the following optimization problem instead of \eqref{optim:functional-penalized}:
\begin{align} \label{optim:discrete-penalized}
    \begin{split}
        \Minimize_{\mathbf{B}} \quad &\sum_{i = 1}^n \frac{1}{M} \sum_{m = 1}^M \bigg( Q_{Y_i}(u_m) - \overline{Q}_Y(u_m) - \sum_{j = 1}^p  b_{j, m} (X_{i, j} - \bar{X}_j) \bigg)^2 + \lambda \mathcal{J}(\mathbf{B})\\  
        \textrm{subject to} \quad &\mathrm{rank}(\mathbf{B}) = r,
    \end{split}
\end{align}
where $\lambda > 0$ is a tuning parameter, $\mathcal{J}(\mathbf{B})$ is a penalty function that regularizes the behavior of the estimate $\hat\beta$ over $\mathcal{U}_M$ as $M$ increases.
Letting $\widehat{\mathbf{B}}$ be the solution of \eqref{optim:discrete-penalized}, we approximate $\hat\beta(u_m) = \big( \hat\beta_1(u_m), \ldots, \hat\beta_p(u_m) \big)^\top$ by the $m$-th column of $\widehat{\mathbf{B}}$.
However, directly solving the above optimization problem with respect to the functional object $\beta$ could be challenging.

\begin{remark}
    If the parsimonious representation of $\mathbf{B}$ is of interest, one may consider employing $\mathcal{J}(\mathbf{B}) = \sum_{j = 1}^p M^{-1} \sum_{m = 1}^M |b_{j, m}| = M^{-1} \| \mathrm{vec}(\mathbf{B}) \|_1$.
    One may also consider the fused lasso–type penalty function $\mathcal{J}(\mathbf{B}) = \sum_{j = 1}^p M^{-1} \sum_{m = 1}^M | b_{j, m} - b_{j, m-1} | = M^{-1} \| D \, \mathrm{vec}(\mathbf{B}) \|_1$, where $D \in \mathbb{R}^{p(M-1) \times (pM)}$ is the first-order difference operator defined as in (2.2) of \citep{li2021double-penalized}. 
    This penalty encourages sparsity in the discrete derivatives of the coefficient functions, producing piecewise constant paths $u \mapsto \hat\beta_j(u)$ for each $j = 1, \ldots, p$ that adaptively capture structural changes. 
    Such a formulation is known to yield edge-preserving smoothing and retaining sharp transitions while suppressing noise, well studied in signal processing and statistical estimation literature  \citep{tibshirani2005sparsity, kim2009ell_1, rudin1992nonlinear}. 
\end{remark}

\subsection{Theory} \label{sec:theory}
In this study, we consider the following technical assumptions.
\begin{enumerate}[label=(A\arabic*)]
    \item \label{assump:support} $s \mapsto F(s|\mu)$ is continuous and strictly increasing on $\mathcal{S}$.
    \item \label{assump:sub-gaussian} $\{ \varepsilon(u): u \in [0, 1] \}$ is a sub-Gaussian process, i.e., $\mathit{E} \big[ e^{\lambda \{\varepsilon(u) - \varepsilon(u')\}} \big] \leq \exp\big\{ \frac{1}{2} \lambda^2 |u - u'|^2 \big\}$ for $\lambda \in \mathbb{R}$ and $u, u' \in [0, 1]$. 
    \item \label{assump:cov-kernel} The covariance function of $T$, defined as $C_T(s, s') = \mathrm{Cov}\big( T(s), T(s') \big)$, is continuous on $\mathcal{S} \times \mathcal{S}$, and $\mathrm{Cov}\big( T(s), T(s') \,|\, X = x \big) = C_T(s, s')$ holds for every $x\in [0, 1]^p$. 
\end{enumerate}
We note that Assumptions \ref{assump:support}--\ref{assump:cov-kernel} are sufficient conditions for the assumption (T1) in \citep{petersen2021wasserstein}.
Specifically, \ref{assump:cov-kernel} and \eqref{OT:error-cov} give $\mathrm{Cov}\big( \varepsilon(u), \varepsilon(u') \big) = C_T\big( \alpha(u), \alpha(u') \big)$ with $\alpha(u) = Q(u|\mu) = F^{-1}(s|\mu)$ as discussed in Remark \ref{rmk:quantile-reparam}. 
This indicates that the pseudo-error \(\varepsilon\) has a homogeneous covariance function, which is a widely adopted assumption in the functional regression analysis literature.  
Moreover, \ref{assump:sub-gaussian} entails $\mathrm{Var}\big(\varepsilon(u) - \varepsilon(u')\big) \leq C( u - u' )^2$ for some absolute constant $C > 0$. 
Since $\varepsilon(0) = \varepsilon(1) = 0$ with \eqref{OT:pseudo-error}, it follows that $\sup_{u \in [0, 1]} \mathrm{Var}(\varepsilon(u))$ is bounded. 
Therefore, $\int_\mathcal{S} C_T(s, s) \, \mathrm{d}F(s|\mu) = \int_0^1 C_T\big(\alpha(u), \alpha(u)\big) \, \mathrm{d}u = \int_0^1 \mathrm{Var}\big( \varepsilon(u) \big) \, \mathrm{d}u  < \infty$, and \ref{assump:support} implies $\sup_{s \in \mathcal{S}} \mathrm{Var}\big(T(s)\big) < \infty$.

Next, using the discretized estimator $ \widehat{\mathbf{B}} $, we show how to recover the true coefficient function $\beta^\ast $ as an infinite-dimensional parameter of rank $r$.
To this end, we consider an additional assumption:
\begin{enumerate}[label=(A\arabic*)]
    \setcounter{enumi}{2}
    \item \label{assump:bounded-variation} $\beta^\ast$ is of bounded variation on $[0, 1]$, meaning that  $\mathrm{TV}(\beta^\ast) = \sup_{K \geq 1} \mathrm{V}(\beta^\ast; \mathcal{P}_K) < \infty$, where $\mathrm{V}(\beta^\ast; \mathcal{P}_K) = \sup \big\{ \sum_{k = 1}^K \| \beta^\ast(\xi_k) - \beta^\ast(\xi_{k - 1}) \|_2 : 0 = \xi_0 < \cdots < \xi_K = 1 \big\}$, and $\mathcal{P}_K$ denotes the collection of all partitions of $[0, 1]$ into $K$ sub-intervals.
\end{enumerate}
Total variation regularization was first introduced by Rudin et al. \citep{rudin1992nonlinear} for image recovery and denoising, and has since been widely studied in the fields of image processing and signal reconstruction. More recently, Lin and M\"uller \citep{lin2021total} extended this regularization framework to Hadamard spaces, where theoretical analysis is complicated by the lack of smooth differential structure. In our setting, we impose condition (A3) to ensure that the discretized coefficient matrix preserves the rank of the underlying smooth coefficient function. The formal result is stated as follows.

\begin{prop} \label{prop:rank}
If Assumption \ref{assump:bounded-variation} holds, then $\mathrm{rank}(\beta^\ast) = \mathrm{rank}(\mathbf{B}^\ast)$ with a sufficiently large $M \geq 1$. 
\end{prop}
The proof of the proposition is straightforward: 
Assuming that $\mathrm{rank}\big( \{ \beta^\ast(u) \in \mathbb{R}^p : u \in [0,1] \} \big) = r$, the image of $\beta^\ast$ is contained by an $r$-dimensional linear subspace of $\mathbb{R}^p$, and there exist points $v_1, \ldots, v_r \in [0,1]$ such that the vectors $\beta^\ast(v_1), \ldots, \beta^\ast(v_r)$ are linearly independent in $\mathbb{R}^p$.
Since $\beta^\ast$ is of bounded variation, it is uniformly continuous on $[0,1]$. 
That is, for any $\delta > 0$, there exists $\eta > 0$ such that $\| \beta^\ast(u) - \beta^\ast(v) \|_2 < \delta$ whenever $|u - v| < \eta$.
We note that the evaluation grid $\mathcal{U}_M$ used in the optimization problem \eqref{optim:discrete-penalized} becomes dense in $[0,1]$ as $M \to \infty$.
This enables to find a sufficiently large $M \geq 1$ such that $\| \beta^\ast(u_{m_k}) - \beta^\ast(v_k) \|_2 < \delta$ for all $k = 1, \ldots, r$ with a subset of grid points $\{ u_{m_k} \in \mathcal{U}_M : k = 1, \ldots, r \}$.
Let $\mathbf{G}_\delta^\ast$ denote the $r \times r$ Gram matrix formed by the vectors $\{ \beta^\ast(u_{m_k}) : k = 1, \ldots, r \}$. 
By the continuity of the determinant, we have $\det(\mathbf{G}_\delta^\ast) > 0$ for sufficiently small $\delta > 0$.
This implies that $\beta^\ast(u_{m_1}), \ldots, \beta^\ast(u_{m_r})$ are also linearly independent, and we have $\mathrm{rank}(\mathbf{B}^\ast) \geq r$. 
On the other hand, since the column space of $\mathbf{B}^\ast$ is contained by the image of $\beta^\ast$, which spans an $r$-dimensional subspace, we have $\mathrm{rank}(\mathbf{B}^\ast) \leq r$. 
Finally, we conclude that $\mathrm{rank}(\mathbf{B}^\ast) = r$.
Based on the above arguments, the rate of convergence of the proposed estimator is given as follows:
\begin{thm} \label{thm:consistency}
    Suppose the above Assumptions \ref{assump:support}--\ref{assump:bounded-variation}, along with Conditions \ref{condition:regularization} and \ref{condition:RSC} in Appendix \ref{sec:conditions}, hold.
    Let $\mathcal{M}_r = \{ \mathbf{B} \in \mathbb{R}^{p \times M} : \mathrm{rank}(\mathbf{B}) = r \}$ be the set of $p \times M$ matrices with rank $r < \min\{p, M\}$. 
    Also, let $\mathbf{B}^\ast \in \mathcal{M}_r$ be the true coefficient matrix associated with the latent model \eqref{model:quantile-reparam}, where the $m$-th column of $\mathbf{B}^\ast$ consists of $\beta^\ast(u_m) \in \mathbb{R}^p$ for $m = 1, \ldots, M$.
    Then, we have
    $\big\| \widehat{\mathbf{B}} - \mathbf{B}^\ast \big\|_F = O_P\big( \sqrt{r (p + M) / n} + \sqrt{\lambda  p M / n} \big)$.
\end{thm}
Let $\widetilde{\mathbf{B}}$ denote the estimator obtained from \eqref{optim:discrete-penalized} without imposing the low-rank constraint. Its rate of convergence is given by $\big\| \widetilde{\mathbf{B}} - \mathbf{B}^\ast \big\|_F = O_P\big( \sqrt{p M / n} + \sqrt{\lambda p M / n} \big)$, which leads to a substantial loss in estimation efficiency compared to the proposed low-rank estimator $\widehat{\mathbf{B}}$.
A similar result was previously established by \citep{huang2024framework} in the context of low-rank regularization for matrix models. Here, we extend the analysis to the setting of global Fr\'echet regression, where the parameter of interest is infinite-dimensional, as detailed below.

\begin{thm} \label{cor:integral}
    Suppose the same conditions of Theorem \ref{thm:consistency}.
    Let $u \mapsto \hat\beta^\ast(u)$ be the linear interpolation of $\hat\beta(u_1), \ldots, \hat\beta(u_M)$, the column space of the solution of \eqref{optim:discrete-penalized}. 
    Then, we have $\| \hat\beta^\ast - \beta^\ast \|_2 = O_P\big( \sqrt{r (p + M) / n} + \sqrt{\lambda  p M / n} \big) + O(M^{-1/2})$.
\end{thm}
The above theorem offers theoretical guidance on selecting the regularization parameter $\lambda$ and the number of evaluation grid points $M \geq 1$.  
For instance, setting $\lambda \asymp (p + M) / (pM)$ balances the trade-off between the mean squared error and the bias introduced by penalization.  
Under this choice, selecting $M \geq (n/p) \vee \sqrt{n}$ ensures that the approximation error from discretization becomes negligible relative to the estimation bias as the sample size $n$ grows.
Lastly, we evaluate the predictive performance of the proposed model in terms of the expected Wasserstein distance.

\begin{thm} \label{cor:mspe}
    Suppose the same conditions of Theorem \ref{thm:consistency} hold.
    Also, assume that all distribution functions in $\mathcal{F}$ have uniformly bounded second moments, i.e, $\sup_{y \in \mathcal{F}} \int_{\mathcal{S}} s^2 \, \mathrm{d}y(s) < \infty$ and that the largest eigenvalue of $\Sigma = \mathrm{Var}(X)$, denoted by $\sigma_{\mathrm{max}}^2(\Sigma)$, is bounded.
    For a random copy $(Y_\textrm{new}, X_\textrm{new})$ of $(Y, X)$, let $\widehat{Y}_\textrm{new}$ be the fitted distributional response defined as $\widehat{Y}_\textrm{new}^{-1}(u) = \overline{Q}_Y(u) + \hat\beta^\ast(u)^\top (X_\textrm{new} - \bar{X})$.
    Define $\sigma_\varepsilon^2 = \int_0^1 \mathrm{Var}\big( \varepsilon(u) \big) \, \mathrm{d}u$.
    Then, we have $\mathit{E} \big[ d_{W_2}(Y_\textrm{new}, \widehat{Y}_\textrm{new}) \,|\, \mathcal{X}_n \big] = O\big(\sigma_\varepsilon\big) + O_P\big( \sqrt{r (p + M) / n} + \sqrt{\lambda  p M / n} \big) + O(M^{-1/2})$. 
\end{thm}

\section{Computational Algorithm}
\label{sec:algorithm}

We present an algorithm for low-rank Fr\'{e}chet quantile regression with regularization, designed to analyze functional data with non-Euclidean response variables in a metric space. This approach enhances stability and interpretability by incorporating a low-rank constraint and penalized regularization. Let $Y_i$ be a response variable in a metric space $(\mathcal{Y}, d)$, and $X_i \in \mathbb{R}^p$ the corresponding predictor variable, for $i = 1, \ldots, n$. The Fr\'{e}chet quantile regression model minimizes a penalized objective function, defined as:
\begin{equation}
\label{algorithm:discrete-penalized}
\min_{\theta} \sum_{i=1}^n w_i d^2(Y_i, Z_i \beta) + \lambda \|\beta\|_1 + \lambda_{\text{fused}} \sum_{j=2}^m \|\beta_j - \beta_{j-1}\|_1,
\end{equation}
where $Z_i \in \mathbb{R}^{m}$ is a design vector, $\theta \in \mathbb{R}^{m \times q}$ is the coefficient matrix, $\lambda$ controls the $\ell_2$-norm penalty, and $\lambda_{\text{fused}}$ governs the fused lasso penalty promoting smoothness, $w_i$ are weights from the low-rank approximation of $\tilde{X}$, $\lambda$ controls the $L_1$ penalty, and $\lambda_{\text{fused}}$ governs the fused lasso penalty promoting smoothness across adjacent coefficients.

\subsection{Parameter Tuning with Smoothed AIC and BIC}
To tune the regularization parameters $\lambda$ and $\lambda_{\text{fused}}$, we utilize smoothed versions of the Akaike Information Criterion (sAIC) and Bayesian Information Criterion (sBIC) as weight-based selection criteria. For a given model with parameters $\hat{\theta}$, we first compute:
\begin{align}
\text{AIC} &= -2 \ell(\hat{\theta}) + 2 \text{df}, \label{eq:aic} \\
\text{BIC} &= -2 \ell(\hat{\theta}) + \log(n) \text{df}, \label{eq:bic}
\end{align}
where $\ell(\hat{\theta})$ is the log-likelihood, $\text{df}$ is the effective degrees of freedom accounting for regularization and low-rank constraints (e.g., the number of non-zero components in $\hat{\theta}$), and $n$ is the sample size. The smoothed weights are then derived as:
\begin{align}
\text{sAIC}_k &= \frac{\exp(-\text{AIC}_k / 2)}{\sum_{m=1}^M \exp(-\text{AIC}_m / 2)}, \label{eq:saic} \\
\text{sBIC}_k &= \frac{\exp(-\text{BIC}_k / 2)}{\sum_{m=1}^M \exp(-\text{BIC}_m / 2)}, \label{eq:sbic}
\end{align}
where $k$ indexes a candidate model over a grid of $(\lambda, \lambda_{\text{fused}})$ pairs, and $M$ is the total number of models. These weights, inspired by \citep{buckland1997model}, balance model fit and complexity, favoring parameter values that optimize predictive accuracy and interpretability. The optimal $(\lambda, \lambda_{\text{fused}})$ pair is selected by maximizing the sAIC or sBIC weight across the grid.

\subsection{Algorithm for Estimating Regularized Fr\'{e}chet Quantile Regression Coefficients}

Estimating the coefficients of our regularized Fr\'echet quantile regression model, as defined in Equation~\eqref{algorithm:discrete-penalized}, involves optimizing a non-smooth objective with $L_1$ and fused lasso penalties, along with a low-rank constraint. The algorithm iterates between a proximal gradient step to handle $L_1$ and fused lasso penalties and a low-rank projection via singular value decomposition (SVD) to enforce rank constraint $r$. For a response $Y_i \in \mathcal{Y}$, observed covariates $\tilde{X}_i \in \mathbb{R}^p$, and design matrix $Z_i \in \mathbb{R}^{m}$ derived from $\tilde{X}_i$, we model the predicted response as $f(\tilde{X}_i; \beta) = Z_i \beta$, where $\beta \in \mathbb{R}^{m \times q}$. 

\begin{algorithm}[ht]
\caption{Fr\'echet Quantile Regression Algorithm}
\label{alg:frechet-quantile-reg}
\begin{algorithmic}[1]
\Require Covariate matrix $\tilde{X} \in \mathbb{R}^{n \times p}$, response $Y \in \mathcal{Y}^n$, design matrix $Z \in \mathbb{R}^{n \times m}$, regularization parameters $\lambda, \lambda_{\text{fused}}$, rank $r$, initial $\beta^{(0)}$, tolerance $\epsilon$, max iterations $T$.
\Ensure Fitted coefficient matrix $\hat{\beta}$.
\State Initialize $\beta^{(0)}$ with small random values.
\For{\textbf{each iteration} $t = 0, 1, \dots, T-1$}
    \State Compute the smooth loss gradient:
    \[
    \nabla_{\text{smooth}} = \frac{1}{n} \sum_{i=1}^n w_i Z_i^T \nabla d^2(Y_i, Z_i \beta^{(t)}),
    \]
    where $\nabla d^2$ is the metric-specific subgradient (e.g., computed via Fr\'{e}chet derivatives).
    \State Perform proximal gradient step for $L_1$ penalty:
    \[
    \beta_{\text{temp}} = \text{prox}_{\alpha \lambda \|\cdot\|_1} (\beta^{(t)} - \alpha \nabla_{\text{smooth}}),
    \]
    where $\text{prox}_{\alpha \lambda \|\cdot\|_1}(v) = \text{sign}(v) \max(|v| - \alpha \lambda, 0)$ is applied element-wise.
    \State Apply proximal operator for fused lasso:
    \[
    \beta^{(t+1/2)} = \text{prox}_{\alpha \lambda_{\text{fused}} \sum_{j=2}^m \|\beta_j - \beta_{j-1}\|_1} (\beta_{\text{temp}}),
    \]
    using an efficient solver.
    \State Project onto low-rank manifold:
    \[
    [U, D, V] \leftarrow \text{SVD}(\beta^{(t+1/2)}), \quad \beta^{(t+1)} \leftarrow U[:, 1:r] D[1:r, 1:r] V[:, 1:r]^T.
    \]
    \State Check convergence: if $\frac{\| \beta^{(t+1)} - \beta^{(t)} \|_F}{\| \beta^{(t)} \|_F} < \epsilon$, break.
    \State Adjust step size $\alpha$ via backtracking line search if needed.
\EndFor
\Return $\hat{\beta} = \beta^{(t+1)}$.
\end{algorithmic}
\end{algorithm}

The algorithm initializes $\theta$ randomly and iterates until convergence or a maximum of $T$ steps. In each iteration, it computes the mean squared error (MSE) gradient and the fused lasso gradient, updates $\theta$ via gradient descent, and enforces a low-rank structure using singular value decomposition (SVD). The rank $r$ is a tuning parameter, typically chosen based on data characteristics or cross-validation. Convergence is assessed using a gradient norm threshold, and the learning rate $\alpha$ may be adapted dynamically (e.g., via line search) to ensure stability.
This approach efficiently handles high-dimensional functional data, with regularization enhancing sparsity and smoothness while the low-rank constraint improves computational tractability and interpretability.

\section{Simulation Study} \label{sec:simulation}

\paragraph{Quantile warping}
We first conducted a numerical simulation to illustrate the theoretical framework presented in Section \ref{sec:theory}.
Let $Z = (Z_1, \ldots, Z_p)^\top$ be a random vector following a multivariate normal distribution, where $\mathit{E} Z_j = 0$ and $\mathrm{Cov}(Z_j, Z_k) = 0.9^{|j - k|}$ for some $\rho \in (-1, 1)$. 
Denoting the cumulative distribution function of $N(0, 1)$ by $\Phi$, we set $X = \big( \Phi(Z_1), \ldots, \Phi(Z_p) \big)^\top$. 
This ensures that $X$ is compactly supported on $[0, 1]^p$ and satisfies $\mathit{E} X = (\mu_1, \ldots, \mu_p)^\top$ with $\mu_j = 0.5$ for all $j = 1, \ldots, p$.
We drew a random sample $\big\{(X_{i, 1}, \ldots, X_{i, p}): i = 1, \ldots, n \big\}$ of $X$.
We designed a simulation scenario where the latent model \eqref{model:quantile-reparam} as 
\begin{align} \label{sim:model-warping}
    Q(u | X_i) = u +  \sum_{j = 1}^p \frac{1}{p} \sum_{k = 0}^r c_{j, k} \gamma_k(u) (X_{i, j} - \mu_j) \quad (u \in [0, 1]) 
\end{align}
for $i = 1, \ldots, n$, where the coefficient vector $c_j = (c_{j, 0}, \ldots, c_{j, r})^\top$ are non-negative constants satisfying $\sum_{k = 0}^r c_{j, k} = 1$ and $\{\gamma_k: k = 0, 1\ldots, r \}$ is a set of linearly independent functions satisfying $\gamma_k(0) = \gamma_k(1) = 0$.
In our simulation study, the coefficient vectors \( c_1, \ldots, c_p \) were independently generated from a \( \mathrm{Dirichlet}(1, \ldots, 1) \) distribution and then held fixed throughout the simulation. To normalize the magnitude of covariate effects as \( p \) increased with the sample size \( n \), we applied the multiplication factor \( p^{-1/2} \).
The latent model \eqref{sim:model-warping} specifies the low-rank representation of the coefficient function associated with \( X_i \) as $\beta_j(u) = p^{-1/2} \sum_{k = 0}^r c_{j, k} \gamma_k(u)$ 
for $j = 1, \ldots, p$. 
To ensure that \( u \mapsto Q(u|X_i) \) remains monotone increasing, we defined \( \gamma_k(u) = \Psi(u; k + 1, r - k + 1) - u \), where \( \Psi(\cdot; k_1, k_2) \) denotes the cumulative distribution function of a \( \mathrm{Beta}(k_1, k_2) \) distribution.
Finally, using the quantile warping approach described in \eqref{OT:pseudo-error}, the response quantile function \( Q_{Y_i} \) was generated as  
\begin{align} \label{sim:respose-warping}
    Q_{Y_i}(u) = \sum_{k = 0}^{50} w_{i, k} \Psi\big(Q(u | X_i); k + 1, 50 - k + 1\big) \quad (u \in [0, 1])
\end{align}
for $i = 1, \ldots, n$, where \( w_i = (w_{i, 1}, \ldots, w_{i, 50})^\top \sim \mathrm{Dirichlet}(1, \ldots, 1) \).
As a result, \( Q_{Y_i} \) can be interpreted as the quantile warping of \( Q(u | X_i) \), where the individual coefficient functions \( \beta_j(u) \) of rank \( r \) represent the increase in \( Q(u|X_i) \) relative to the marginal Fr\'echet mean \( \mathrm{Uniform}(0, 1) \) for a unit increase in \( X_{i, j} \).

\begin{figure}[!th]
    \centering
    \includegraphics[width=0.325\linewidth]{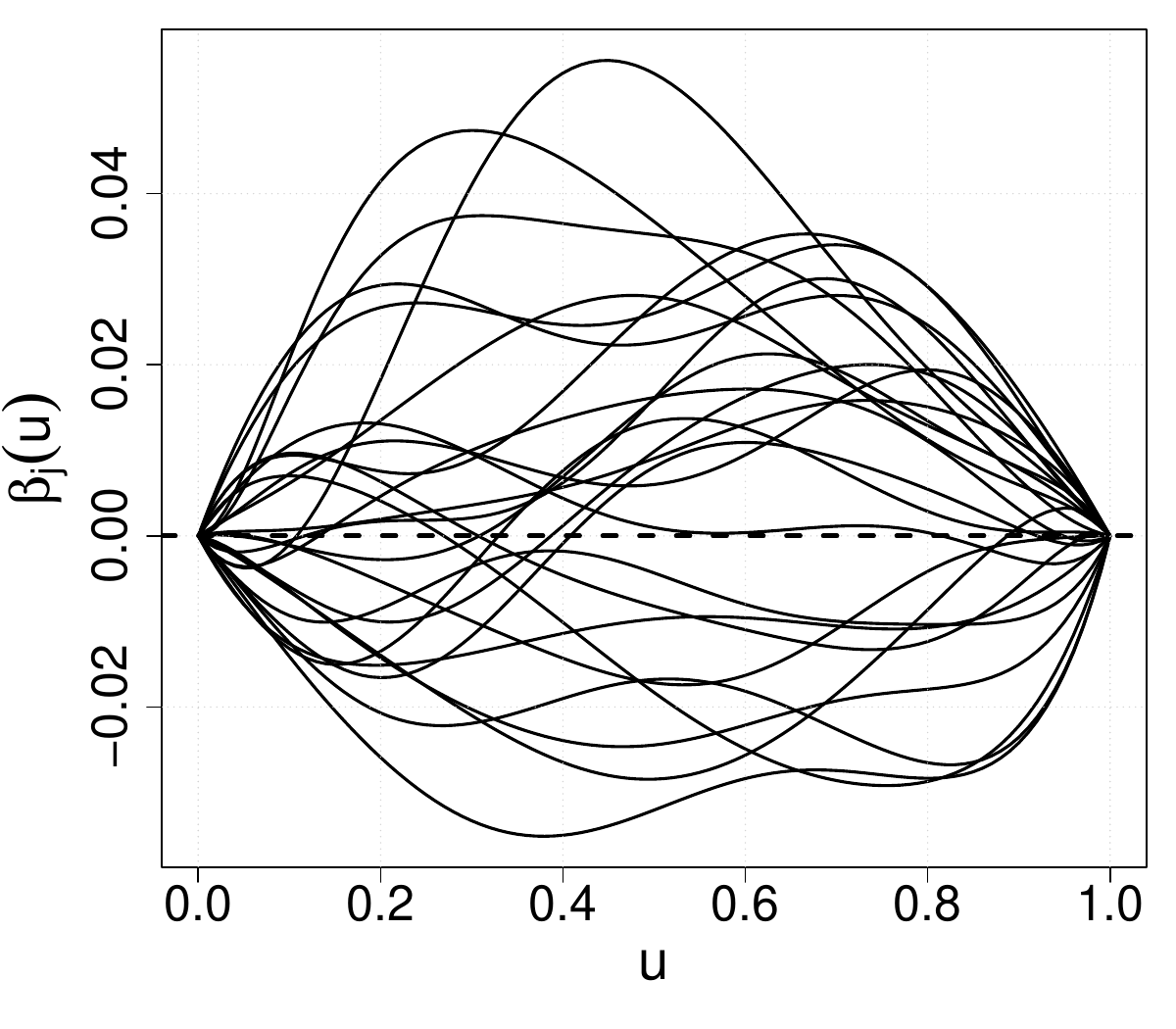}
    \includegraphics[width=0.325\linewidth]{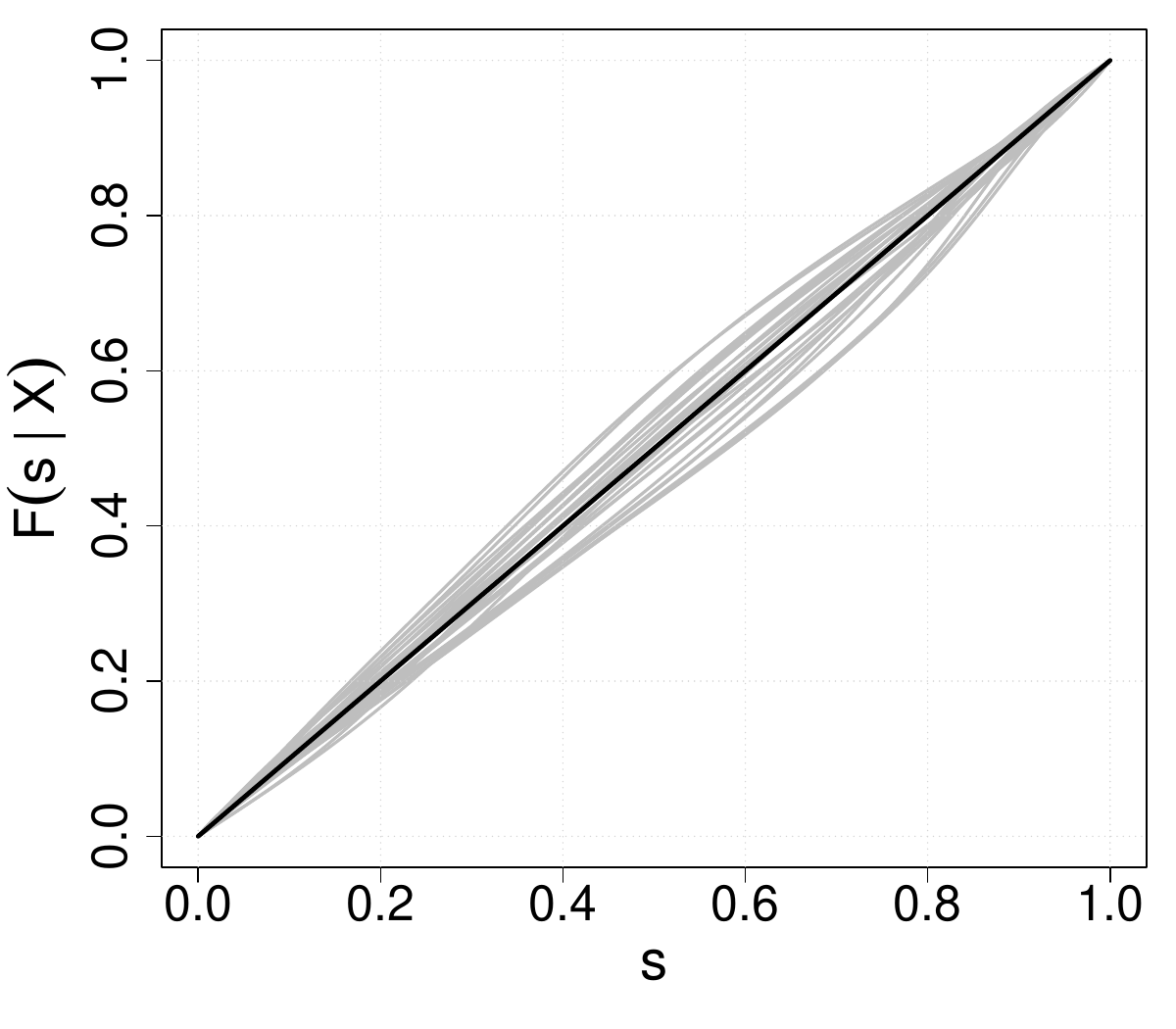}
    \includegraphics[width=0.325\linewidth]{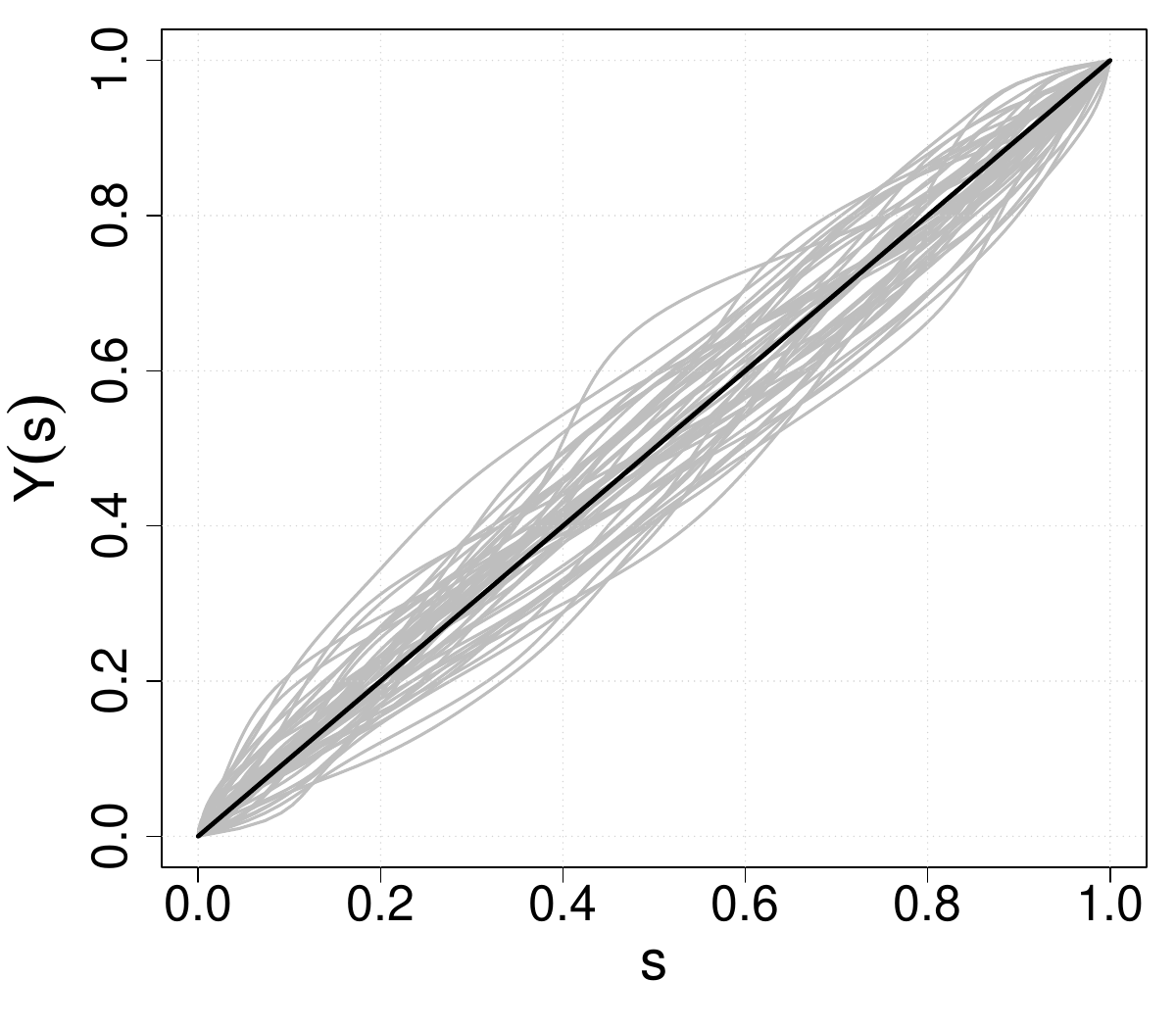}\\
    \includegraphics[width=0.325\linewidth]{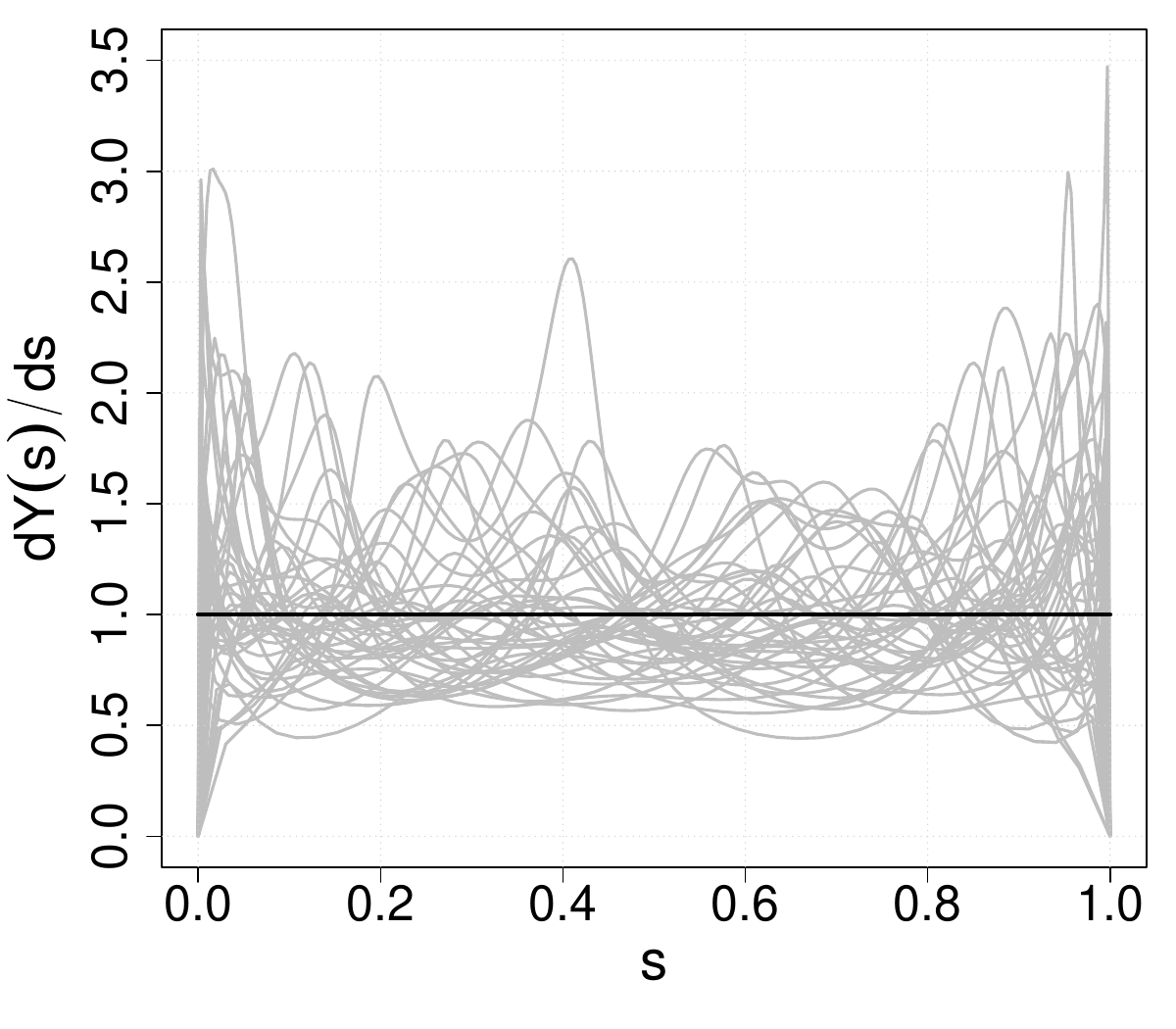}
    \includegraphics[width=0.325\linewidth]{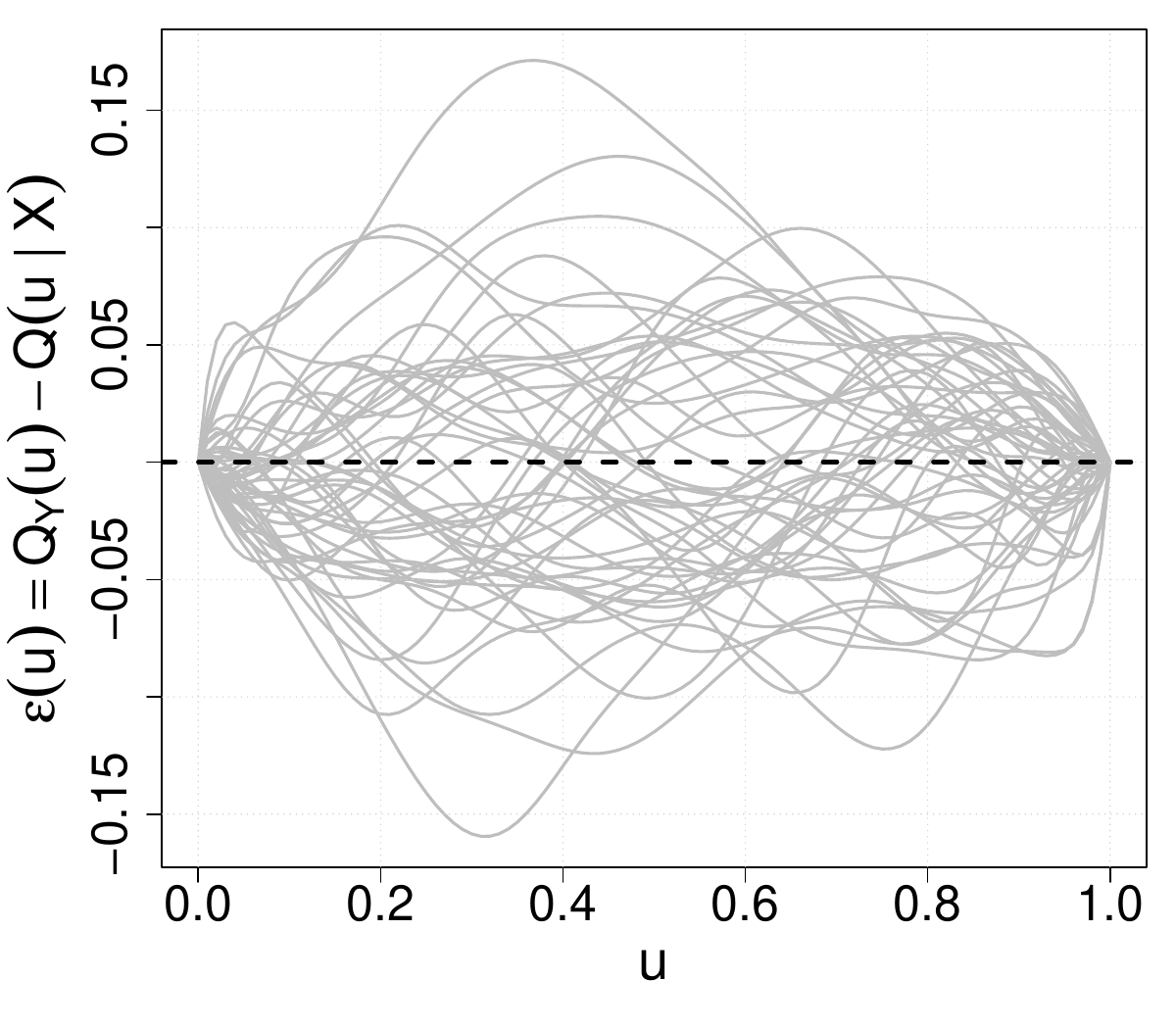}
    \includegraphics[width=0.325\linewidth]{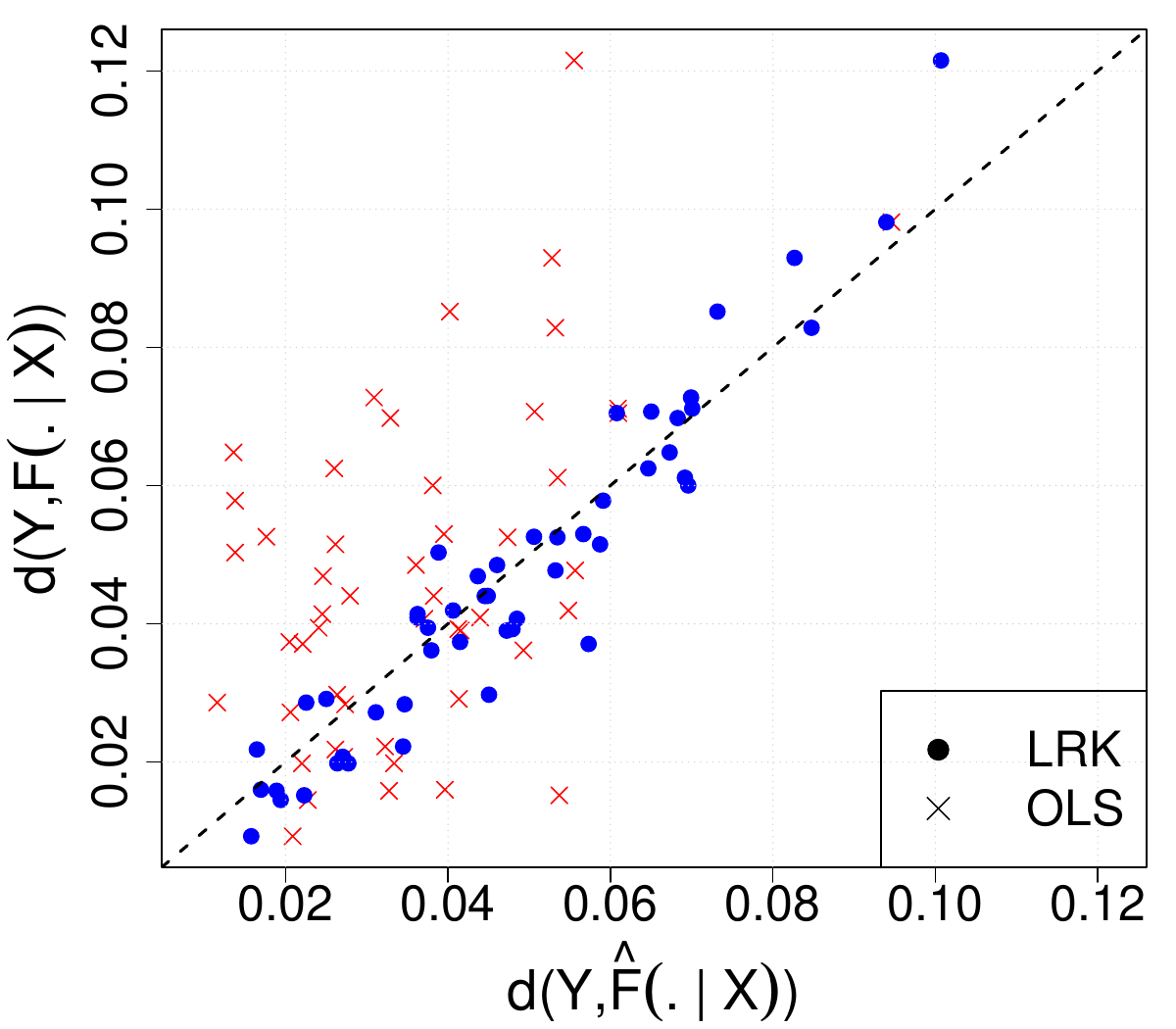}
    \caption{Illustration of simulation under the quantile warping with $n = 50$, $p = 25$, and $r = 10$: (top left) regression coefficient functions $\beta_j(u) = p^{-1} \sum_{k = 0}^r c_{j, k} \gamma_k(u)$ in \eqref{sim:model-warping} ; (top middle) latent models $F(s|X_i) = Q^{-1}(s|X_i)$; (top right and bottom left) response distribution functions $Y_i(s)$ and density functions $Y_i'(s)$; (bottom middle) the pseudo-errors $\varepsilon_i(u)$ associated with the quantile warping \eqref{sim:respose-warping}; (bottom right) the comparison of 2-Wasserstein distance residual-error plots for the proposed low-rank regularization (LRK) method and the ordinary least squares (OLS).}
    \label{fig:sim-illustration-KH}
\end{figure}

To evaluate the finite-sample performance of the proposed method, we generated Monte Carlo (MC) repetitions \( \mathcal{X}_n^{(1)}, \ldots, \mathcal{X}_n^{(B)} \) of a random sample \( \mathcal{X}_n = \big\{ (Y_i, X_i): i = 1, \ldots, n \big\} \) and computed the estimate \( \hat{\beta}^{(b)} \) of \( \beta \) using \( \mathcal{X}_n^{(b)} \) for each \( b = 1, \ldots, B \). 
We then decomposed the mean squared error (MSE) of \( \hat{\beta} \), given by \( \mathrm{MSE}(\hat{\beta}) = B^{-1} \sum_{b = 1}^B \int_0^1 \big\| \hat{\beta}^{(b)}(u) - \beta(u) \big\|_2^2 \, \mathrm{d}u\) 
into its bias and variance components as
\begin{align} \nonumber
    \begin{split}
        \mathrm{Bias}\big(\hat\beta\big)^2
            = \int_0^1 \Big\| \beta(u) - \bar{\hat\beta}(u) \Big\|_2^2 \, \mathrm{d}u
        \quad \text{and} \quad
        \mathrm{Var}\big(\hat\beta\big)
            = \frac{1}{B} \sum_{b = 1}^{B} \int_0^1 \Big\|\hat\beta_j^{(b)}(u) - \bar{\hat\beta}_j(u) \Big\|_2^2 \, \mathrm{d}u,
    \end{split}
\end{align}
where $\bar{\hat\beta}(u) = B^{-1} \sum_{b = 1}^B \hat\beta^{(b)}(u)$ is the MC average of estimates.
Also, we define the mean prediction error as
\begin{align} \nonumber
    \begin{split}
        \mathrm{PE}_{in}
            = \frac{1}{B} \sum_{b = 1}^{B} \frac{1}{n} \sum_{i = 1}^n d_{W_2}\big( Y_i^{(b)}, \widehat{Y}_i^{(b)} \big)
        \quad \text{and} \quad
        \mathrm{PE}_{out}
            = \frac{1}{B} \sum_{b = 1}^{B} \frac{1}{N} \sum_{i = 1}^N d_{W_2}\big( Y_i^\ast, \widehat{Y}_i^{(b, \ast)} \big).
    \end{split}
\end{align}
Here, \( \widehat{Y}_i^{(b)} = \widehat{F}^{(b)}(\cdot|X_i^{(b)}) \) represents the in-sample prediction of \( Y_i^{(b)} \), where \( \widehat{F}^{(b)}(\cdot|x) \) is estimated using \( \mathcal{X}_n^{(b)} \) at the evaluation point \( X = x \). Similarly, \( \widehat{Y}_i^\ast = \widehat{F}^{(b)}(\cdot|X_i^\ast) \) represents the out-of-sample prediction of \( Y_i^\ast \) evaluated at \( X = X_i^\ast \), where \( \mathcal{X}_N^\ast = \{ (Y_i^\ast, X_i^\ast): i = 1, \ldots, N \} \) is a random sample of \( (Y, X) \) that is independent of \( \mathcal{X}_n^{(1)}, \ldots, \mathcal{X}_n^{(b)} \). 

\begingroup
\setlength{\tabcolsep}{6pt} 
\renewcommand{\arraystretch}{1.15} 
\begin{table}[!t]
    \small
    \centering
    \caption{Monte Carlo simulation of size $B = 100$. The tuning parameters are chosen by using sAIC}
    \label{tab:sim-MC-KH}
    \begin{tabular}{cc ccccc ccccc }
    \hline
    \multirow{1.75}{*}{Rank}     &   \multirow{1.75}{*}{Design}   &    
    \multicolumn{5}{c}{\multirow{1.25}{*}{Low-rank regularization}}  & \multicolumn{5}{c}{\multirow{1.25}{*}{Ordinary least squares}}\\
    \cmidrule(lr){3-7} \cmidrule(lr){8-12}
    \multirow{0.25}{*}{($r$)} &   \multirow{0.25}{*}{($n = 2p$)}    &   $\sqrt\mathrm{MSE}$ &  $\mathrm{Bias}$   &   $\sqrt\mathrm{Var}$  &    $\mathrm{PE}_{in}$  &   $\mathrm{PE}_{out}$  &  $\sqrt\mathrm{MSE}$ &  $\mathrm{Bias}$   &   $\sqrt\mathrm{Var}$  &    $\mathrm{PE}_{in}$  &   $\mathrm{PE}_{out}$ \\
    \hline
        \multirow{4}{*}{5}    
        &   50       &  0.108   &   0.107   &   0.015   &   0.045   &   0.049   &  0.578   &   0.051   &   0.575   &   0.032   &   0.069   \\
        &   100      &  0.104   &   0.103   &   0.014   &   0.045   &   0.052   &  0.539   &   0.056   &   0.536   &   0.032   &   0.067   \\
        &   200      &  0.111   &   0.110   &   0.013   &   0.045   &   0.050   &  0.533   &   0.053   &   0.530   &   0.033   &   0.066   \\
        &   400      &  {\color{black} 0.103}   &  {\color{black} 0.102}   &  {\color{black} 0.013}   &  {\color{black} 0.044}   &  {\color{black} 0.052}   &  {\color{black} 0.526}   &  {\color{black} 0.100}   &  {\color{black} 0.517}   &  {\color{black} 0.032}   &  {\color{black} 0.066}   \\
    \hline
        \multirow{4}{*}{10}    
        &   50       &  0.090   &   0.086   &   0.024   &   0.043   &   0.050   &   0.577  &   0.052   &   0.574   &   0.032   &   0.069   \\
        &   100      &  0.097   &   0.095   &   0.020   &   0.043   &   0.050   &  0.539   &   0.055   &   0.536   &   0.032   &   0.067   \\
        &   200      &  {\color{black} 0.094}   &  {\color{black} 0.091}   &  {\color{black} 0.020}   &  {\color{black} 0.044}   &  {\color{black} 0.049}   &  {\color{black} 0.531}   &  {\color{black} 0.113}   &  {\color{black} 0.519}   &  {\color{black} 0.032}   &  {\color{black} 0.067}   \\
        &   400      &  {\color{black} 0.087}   &  {\color{black} 0.085}   &  {\color{black} 0.019}   &  {\color{black} 0.043}   &  {\color{black} 0.050}   &  {\color{black} 0.527}   &  {\color{black} 0.101}   &  {\color{black} 0.517}   &  {\color{black} 0.032}   &  {\color{black} 0.066}   \\
    \hline
        \multirow{4}{*}{20}    
        &   50       &  {\color{black} 0.085}   &  {\color{black} 0.071}   &  {\color{black} 0.046}   &  {\color{black} 0.040}   &  {\color{black} 0.049}   &  {\color{black} 0.573}   &  {\color{black} 0.103}   &  {\color{black} 0.564}   &  {\color{black} 0.033}   &  {\color{black} 0.069}   \\
        &   100      &  {\color{black} 0.079}   &  {\color{black} 0.074}   &  {\color{black} 0.028}   &  {\color{black} 0.042}   &  {\color{black} 0.049}   &  {\color{black} 0.542}   &  {\color{black} 0.100}   &  {\color{black} 0.532}   &  {\color{black} 0.032}   &  {\color{black} 0.067}   \\
        &   200      &  {\color{black} 0.075}   &  {\color{black} 0.069}   &  {\color{black} 0.029}   &  {\color{black} 0.042}   &  {\color{black} 0.049}   &  {\color{black} 0.530}   &  {\color{black} 0.113}   &  {\color{black} 0.518}   &  {\color{black} 0.032}   &  {\color{black} 0.067}   \\
        &   400      &  {\color{black} 0.071}   &  {\color{black} 0.065}   &  {\color{black} 0.028}   &  {\color{black} 0.041}   &  {\color{black} 0.049}   &  {\color{black} 0.529}   &  {\color{black} 0.102}   &  {\color{black} 0.519}   &  {\color{black} 0.032}   &  {\color{black} 0.066}   \\
    \hline
    \end{tabular}
\end{table}
\endgroup


\paragraph{Functional linear regression for quantile functions}
We conducted a simulation study to evaluate the performance of the proposed model. The results demonstrate the efficacy of the quantile low-rank Fr\'echet Regression in capturing the underlying structure and providing accurate quantile estimates.
Here is the dimension setting: 
The sample size $n = 200$, $p=180$, $M = 100$, $u=(0,0.1,0.2,\ldots,0.9,1)$, $u=(-5,-4.98,\ldots,4.98,5)$.
The quantiles were simulated as following:
\[
Q_{i,j} = \mu_j + \sum_{k=1}^{K} u_{i,k} \sigma_{k,j} + \epsilon_{i,j}, \quad i = 1, \ldots, n, \quad j = 1, \ldots, m
\]
where $\mu_j$ is the mean parameter for the $j$-th quantile, $u_{i,k}$is the latent factors affecting the quantiles, $\sigma_{k,j}$ are the weights representing the relationship between the latent factors and quantiles, and $\epsilon_{i,j}$:are independent error terms.
The predictor matrix $X \in \mathbb{R}^{n \times p}$ was generated randomly, and the response $Y$ was constructed using the formula above.
The ordinary method utilized the standard regularization approach without explicitly incorporating the Fr\'echet distance into the optimization. The regularized regression was fitted using
\[
\hat{Q} = Z \hat{\theta} + \mathbf{1} \mu^\top
\]
where $\hat{\theta}$ was estimated using a penalized least squares approach.
The simulations were conducted using the following setup:
the rank for low-rank approximation: $r = 2$ and the $L_1$ regularization parameters: $\lambda = 0.01$
Tables~\ref{tab:table_1} and \ref{tab:table_2} show the average RMSE and its standard deviation of residuals (in parentheses).
The proposed regularized low-rank Fr\'echet quantile regression method outperformed the empirical method in terms of RMSE, demonstrating its effectiveness in leveraging low-rank structures and Fr\'echet distances for quantile regression.

\begin{table}[htbp]
\centering
\caption{Simulation data analysis comparison. The root mean squared error (RMSE) and standard deviation in the parentheses when $n=2p$ and rank as 2. The proposed method has similar training and testing results while the ordinary method tends to overfit the training data.}
\label{tab:table_1}
\begin{tabular}{lcccc}
\hline
Method &Proposed-Training & Ordinary-Training &Proposed-Testing & Ordinary-Testing\\
\hline
($n=100=2p$)& 
0.5318 (0.3139)&
0.4284 (0.2574)&
0.5415 (0.3495)&
0.8658 (0.4870)\\
($n=200=2p$)& 
0.4753 (0.2944)&
0.3601 (0.2201)&
0.5562 (0.3209)&
0.7555 (0.4304)\\
($n=400=2p$)&
0.5575 (0.3277)&
0.3952 (0.2444)&
0.5552 (0.3568)&
0.8293 (0.4874)\\
($n=800=2p$)&
0.5111 (0.3156)&
0.3803 (0.2387)&
0.5213 (0.3332)&
0.7487 (0.4350)\\
($n=1000=2p$)&
0.5554 (0.3281)&
0.4005 (0.2401)&
0.5331 (0.2809)&
0.7800 (0.4399)\\
\hline
\end{tabular}

\end{table}
\begin{table}[htbp]
\centering
\caption{Simulation data analysis comparison. The root mean squared error (RMSE) and standard deviation in the parentheses when $n=p$ with rank as 2. The proposed method has similar training and testing results while the ordinary method is not applicable in the high-dimensional data.}
\label{tab:table_2}
\begin{tabular}{lcc}
\hline
Method &Proposed-Training &Proposed-Testing \\
\hline
($n=p=100$)& 
0.5098 (0.3036)&
0.5387 (0.3364)\\
($n=p=200$)& 
0.4380 (0.2334)&
0.5122 (0.2843)\\
($n=p=400$)&
0.4484 (0.2477)&
0.5444 (0.3077)\\
($n=p=800$)&
0.4990 (0.2824)&
0.5633 (0.3714)\\
($n=p=1000$)&
0.5226 (0.3435)&
0.5374 (0.3320)\\
\hline
\end{tabular}
\end{table}


\section{Real Data Applications} \label{sec:data}
We applied the proposed regularized low-rank Fr\'{e}chet quantile regression (FQR) model to real-world datasets, comparing its performance against functional quantile principal component analysis with regression (FQPCR) \citep{weng2023sparse} using a residual-based framework. Residuals are defined as \(R_i = Y_i - \hat{Q}_\tau(X_i)\), where \(Y_i\) denotes the observed quantile function for the \(i\)-th observation, and \(\hat{Q}_\tau(X_i)\) is the predicted quantile at level \(\tau\) given predictors \(X_i\). This comparison evaluates the residual structure to assess model fit and predictive accuracy across both methods.

\subsection{Application to Global Mortality Data}
We analyzed global mortality data from the United Nations World Population Prospects 2024, comprising abridged life table death counts by age group across 165 countries for 2019 (training) and 2020 (test), adapting a framework from \citep{zhang2023dimension}. The response variable consists of mortality count quantiles over 100 age intervals, scaled to [0, 1]. Predictors—gross domestic product (GDP) per capita, gross value added (GVA) by agriculture, consumer price index (CPI), and health expenditure—were obtained for both years.

The regularized low-rank Fr\'{e}chet quantile regression (FQR) model, with \(\lambda = 0.001\) and \(\lambda_{\text{fused}} = 0.01325\) selected via sAIC, was fitted at rank 2. We compared it to functional quantile principal component analysis (FQPCR) with regression \citep{weng2023sparse}, extracting two components from 2019 quantiles and regressing the second score on the predictors via least squares. For 2020, test scores were predicted, adjusting predictor names, and quantiles reconstructed using training loadings with monotonicity enforced.

Test RMSE for 2020 averaged 0.0385 for FQR and 0.0645 for FQPCR regression (Table~\ref{tab:table_3}). Figure~\ref{fig:mortality_residual_comparison} confirms FQR’s tighter residual distribution, indicating better predictive accuracy. Regression on the Fr\'{e}chet central subspace scores (\(\hat{\beta}^T X\)) produced adjusted R\(^2\) values of 0.464 for FQR versus 0.321 for FQPCR, reinforcing FQR’s superior fit. The GDP per capita coefficient (blue curve, Figure~\ref{fig:mortality_fitted_obs_plots}(d)) mirrors mortality trends, underscoring its explanatory power, while \(\beta_2\) captures age-specific economic influences.

\begin{table}[!t]
\centering
\caption{RMSE for 2020 Mortality Test Data}
\begin{tabular}{lcc}
\hline
Method & FQR & FQPCR \\
\hline
RMSE   & 0.0385 & 0.0645 \\
\hline
\end{tabular}
\label{tab:table_3}
\end{table}

\begin{figure}[htbp]
\centering
\includegraphics[width=0.75\textwidth]{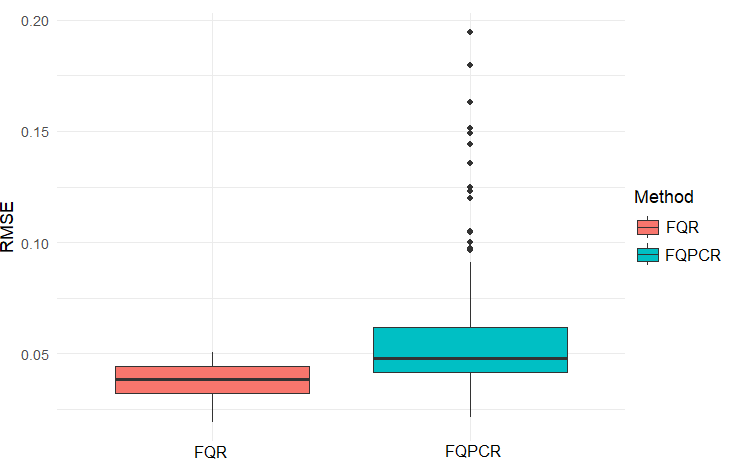}
\caption{Quantile residuals from Fr\'{e}chet quantile regression (FQR) and FQPCR with regression on 2020 mortality test data.}
\label{fig:mortality_residual_comparison}
\end{figure}

\begin{figure}[htbp]
\centering
\includegraphics[width=0.8\textwidth]{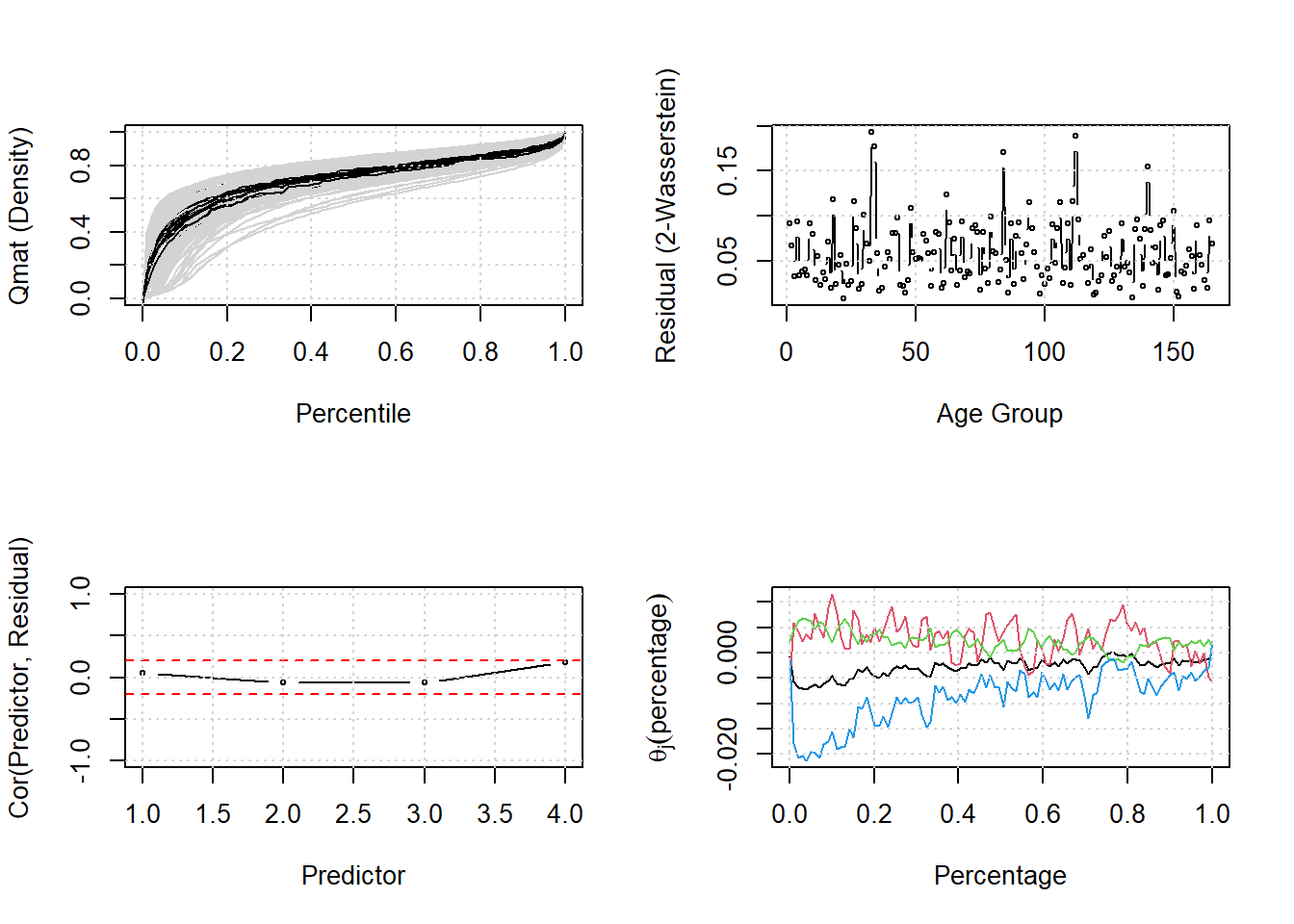}
\caption{Fitted (black) and observed (gray) quantile densities, residuals, correlations, and coefficients from regularized low-rank Fr\'{e}chet quantile regression (rank 2) on mortality data.}
\label{fig:mortality_fitted_obs_plots}
\end{figure}

\subsection{Application to Bike Rental Data}
We applied the proposed regularized low-rank Fr\'{e}chet quantile regression (FQR) model to a bike rental dataset to examine factors influencing hourly rental patterns. This dataset, detailed in \citep{weng2023sparse}, includes normalized bike rental counts as the response variable and six predictors: Holiday (indicator of a public holiday in Washington, D.C.), Working (indicator of neither a weekend nor holiday), Temp (daily mean temperature), Atemp (feels-like temperature), BW (bad weather indicator), and RBW (really bad weather indicator). These predictors capture weather and calendar effects on rental behavior.

The FQR model was fitted with tuning parameters \(\lambda\) and \(\lambda_{\text{fused}}\) selected via smoothed AIC and BIC criteria (sAIC and sBIC) over a grid search. Its performance was compared against Fr\'{e}chet regression with Wasserstein distance (WRI) \citep{petersen2021wasserstein}, functional quantile principal component analysis (FQPCR), and vector quantile regression (VQR) \citep{carlier2016vector}. Table~\ref{tab:WassR_bike} reports root mean squared error (RMSE) on a test set, showing FQR’s superior predictive accuracy with an RMSE of 0.0382, compared to 0.0839 for FQPCR (VQR and WRI results omitted for brevity but follow a similar trend).

\begin{table}[htbp]
\centering
\caption{Root Mean Squared Error (RMSE) for Bike Rental Data on the Test Set}
\begin{tabular}{lcc}
\hline
Method & FQR & FQPCR \\
\hline
RMSE   & 0.0382 & 0.0839 \\
\hline
\end{tabular}
\label{tab:WassR_bike}
\end{table}

Temporal patterns identified by FQR’s first sufficient predictor reveal distinct rental behaviors. On working days, rentals peak twice—between 5 AM and 10 AM, and 3 PM and 8 PM—reflecting commuting trends. On weekends or holidays, a single peak occurs from 10 AM to 8 PM, with maximum activity around 2 PM to 3 PM. The fitted coefficient \(\beta_2\) for the Working predictor, shown as the red curve in Figure~\ref{fig:bike1}(d), effectively captures these dynamics. Figure~\ref{fig:bike1} further illustrates FQR’s fit, depicting observed (gray) and fitted (black) densities, residuals, correlations, and coefficients with a rank-2 model, highlighting its ability to model non-Euclidean responses.
Residual analysis, presented in Figure~\ref{fig:residual_comparison}, compares quantile residuals from FQR and FQPCR on the test set. FQR exhibits lower residual variability, indicating a better fit than competing methods (VQR residuals follow a similar pattern, not shown).
This application underscores FQR’s effectiveness in modeling non-Euclidean responses, revealing temporal rental patterns, and outperforming alternative approaches in predictive accuracy.

\begin{figure}[htbp]
\centering
\includegraphics[width=0.8\textwidth]{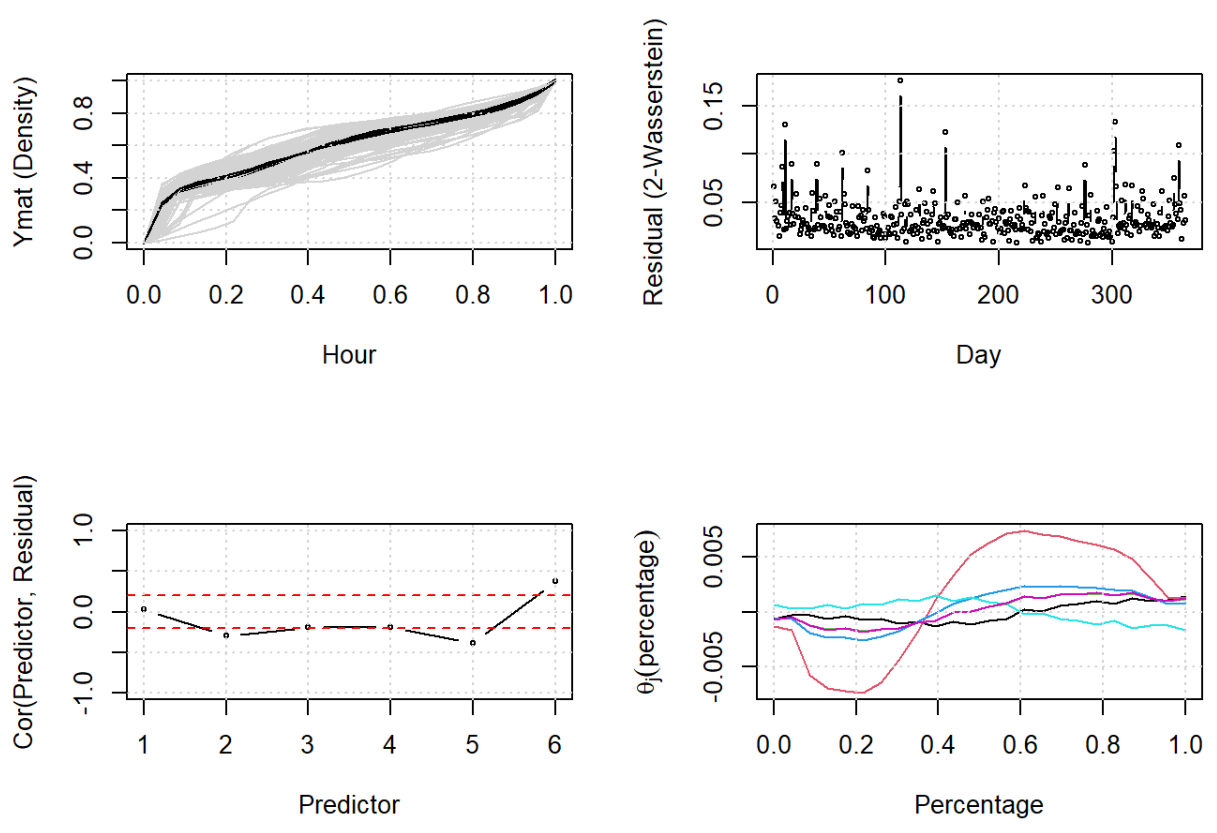}
\caption{Fitted (black) and observed (gray) densities, residuals, correlations, and coefficients from the regularized low-rank Fr\'{e}chet quantile regression model (rank = 2) applied to bike rental data.}
\label{fig:bike1}
\end{figure}

\begin{figure}[htbp]
\centering
\includegraphics[scale=0.75]{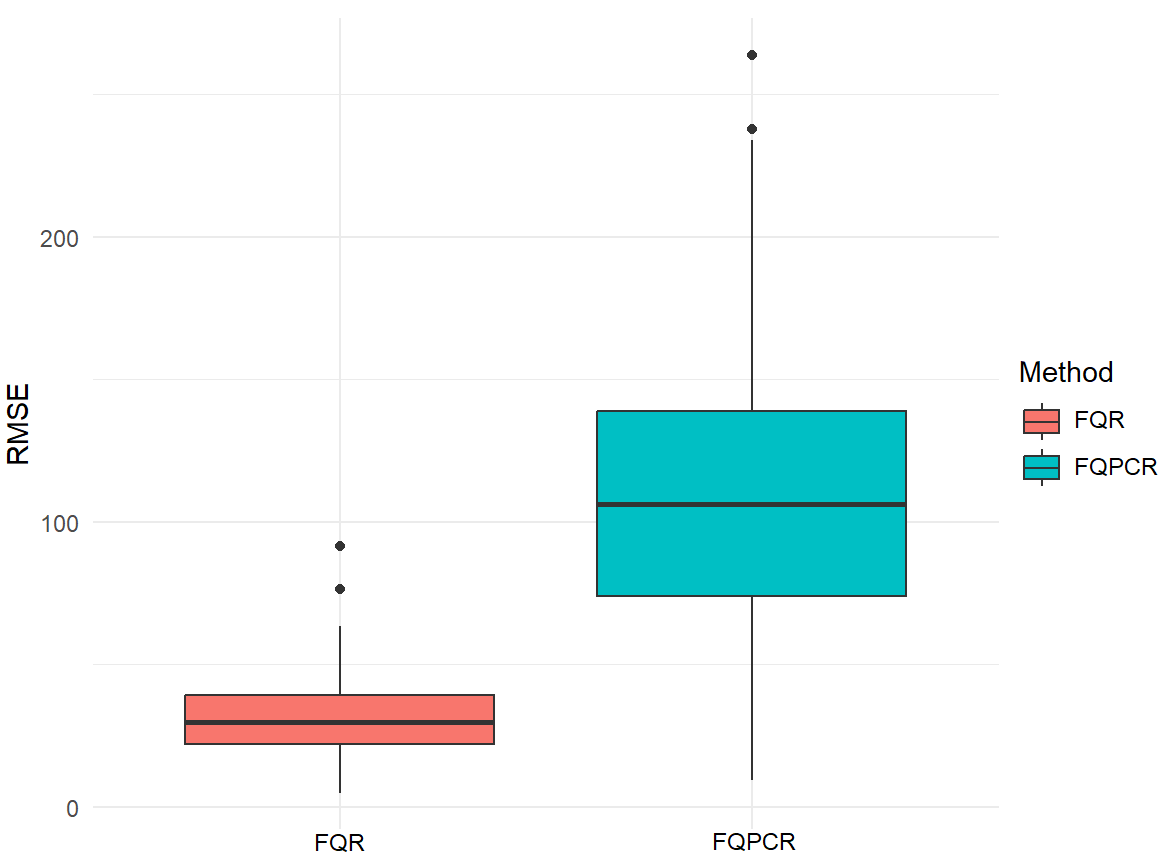}
\caption{Quantile residuals from Fr\'{e}chet quantile regression (FQR) and functional quantile PCA (FQPCR) on the bike rental test set.}
\label{fig:residual_comparison}
\end{figure}

\section{Conclusion}
\label{sec:conclusion}

The low-rank Fr\'echet quantile regression model introduced in this paper offers a principled and effective approach for analyzing distributional response data in high-dimensional settings. Classical Fr\'echet regression methods often suffer from instability or inefficiency when dealing with high-dimensional or collinear covariates, due to the ill-posed nature of inverting sample covariance matrices. Our framework addresses this issue by integrating a latent quantile regression formulation with low-rank constraints and regularization. Specifically, the model transforms the regression problem from the space of distribution functions into the space of quantile functions and postulates a latent functional linear model $\mathrm{E}(Q_Y(u) \mid X = x) = \alpha_0(u) + \beta(u)^\top x$. This structure allows us to define interpretable coefficient functions while respecting the inherent geometry of distributional data. Through a reparameterization, the intercept $\alpha(u)$ is shown to represent the quantile function of the marginal Fr\'echet mean, thereby linking the model directly to the classical notion of mean regression in the Wasserstein space.

To estimate the coefficient function $\beta(u)$ in practice, we propose to discretize the quantile domain over a uniform grid and reformulate the estimation as a matrix optimization problem. This involves recovering a low-rank coefficient matrix $\mathbf{B} \in \mathbb{R}^{p \times M}$ under a penalized least-squares loss, where the penalty function $\mathcal{J}(\mathbf{B})$ enforces desirable properties such as sparsity or smoothness. We incorporate either an $\ell_1$ penalty for coefficient sparsity or a fused lasso penalty to promote piecewise constant behavior across the quantile index, enabling edge-preserving smoothing. This design makes the model highly adaptive and interpretable, while also ensuring numerical stability in high-dimensional or low-sample-size scenarios. Unlike prior work, such as that of Song and Han (2023), our methodology directly accommodates high-dimensional covariates and remains computationally tractable, even when the number of predictors exceeds the sample size.

The theoretical properties of the estimator are rigorously justified using tools from optimal transport and functional regression. By viewing the discrepancy between the true and modeled quantile functions as a pseudo-error process, we show that the model admits an interpretation analogous to classical linear regression, but in a non-Euclidean setting. Under standard assumptions—including sub-Gaussianity of the pseudo-error, bounded second moments, and bounded variation of the true coefficient function—we derive convergence rates for the discretized coefficient matrix $\widehat{\mathbf{B}}$ and its functional extension $\hat{\beta}(u)$. The low-rank estimator achieves a rate of $O_P\big( \sqrt{r(p+M)/n} + \sqrt{\lambda p M/n} \big)$, significantly improving over unconstrained estimators. Moreover, we establish that with a sufficiently fine grid, the discretized estimator recovers the true functional rank, and the approximation error from discretization vanishes at rate $O(M^{-1/2})$. These theoretical guarantees also extend to predictive performance, showing that the expected squared Wasserstein distance between the predicted and true response distributions converges to the irreducible noise level $\sigma_\varepsilon$, up to the same statistical and approximation rates.

Despite these strengths, the proposed method is subject to several limitations from theoretical, methodological, and computational perspectives. Theoretically, the reliance on discretizing the quantile domain introduces a dependence on grid size $M$. A coarse grid may result in significant discretization errors, particularly for complex or heavy-tailed quantile functions, while an overly fine grid increases computational complexity and risks numerical instability. The assumption of a low-rank structure for the coefficient matrix may not hold for datasets with high-rank or complex dependencies, potentially leading to biased estimates. Additionally, the model's linear assumption may be restrictive for capturing non-linear relationships or interactions, limiting its generalizability. Methodologically, the sensitivity to regularization parameters $\lambda$ and $\lambda_{\text{fused}}$ poses challenges, as misspecification can lead to over- or under-regularization, affecting interpretability and performance. From a computational perspective, the non-convexity introduced by the low-rank constraint via singular value decomposition (SVD) complicates convergence guarantees and increases the risk of converging to local minima. The SVD computation, with a cost of $O(\min(pM^2, p^2M))$, becomes a bottleneck for large $p$ or $M$, limiting scalability. Moreover, the evaluation of the Fr\'{e}chet derivative in complex metric spaces can be numerically unstable, particularly for non-smooth distributions or small samples.

Finally, simulation studies and real data applications demonstrate the practical value of the proposed method. In high-dimensional settings where traditional methods such as ordinary least squares or functional principal component regression break down, our estimator maintains both accuracy and interpretability. For example, in analyzing global mortality and bike rental datasets, the model achieves a test RMSE of 0.0385, outperforming existing benchmarks. The ability to impose structural assumptions such as low rank and piecewise smoothness enables the model to capture essential patterns while suppressing noise. This balance of theoretical rigor, computational feasibility, and empirical effectiveness makes the proposed framework a powerful and flexible tool for distributional regression in complex data environments. Future research directions include developing adaptive rank selection procedures, incorporating structured sparsity patterns, and extending the framework to handle more general metric-space responses and covariate types.

\appendix

\section{Appendix}
\subsection{Technical Conditions} \label{sec:conditions}
Below, we have collected technical conditions for the main results.
\begin{enumerate}[label=(C\arabic*)]
    \item \label{condition:regularization} For any $\epsilon > 0$, there exists $C_{\mathcal{J}, p} = C_{\mathcal{J}, p}(\epsilon) > 0$ such that $M^{-1/2} \| \mathbf{B} - \mathbf{B}^\ast \|_F > \epsilon \Rightarrow \big|\mathcal{J}(\mathbf{B}) - \mathcal{J}(\mathbf{B}^\ast) \big| \leq C_{\mathcal{J}, p} M^{-1/2}\| \mathbf{B} - \mathbf{B}^\ast \|_F$ for every $\mathbf{B} \in \mathcal{M}_r$ near $\mathbf{B}^\ast$. 
    \item \label{condition:RSC} For each $\epsilon > 0$, there exits $c_{\mathbf{X}} = c_{\mathbf{X}}(\epsilon) > 0$ such that for any $\mathbf{B} \in \mathcal{M}_r$, $M^{-1/2} \| \mathbf{B} - \mathbf{B}^\ast \|_F \geq \epsilon \Rightarrow n^{-1} \| \widetilde{\mathbf{X}} (\mathbf{B} - \mathbf{B}^\ast)\|_F^2 \geq c_{\mathbf{X}}^2 \| \mathbf{B} - \mathbf{B}^\ast \|_F^2$. 
\end{enumerate}
Condition \ref{condition:regularization} can be regarded as the local Lipschitz continuity of $\mathcal{J}(\mathbf{B})$ near $\mathbf{B}^\ast$, which is useful when the penalty function is non-differentiable near the true $\mathbf{B}^\ast$.
For example, both $\mathcal{J}(\mathbf{B}) = \sum_{j = 1}^p M^{-1} \sum_{m = 1}^M |b_{j, m}|$ and $\mathcal{J}(\mathbf{B}) = \sum_{j = 1}^p M^{-1} \sum_{m = 1}^M | b_{j, m} - b_{j, m-1} |$ with $C_{\mathcal{J}, p}  = C_{\mathcal{J}} \sqrt{p}$ satisfy the condition with some absolute $C_{\mathcal{J}} > 0$.
Condition \ref{condition:RSC} corresponds to the restricted strong convexity \citep{negahban2011estimation, negahban2012unified} of $\mathcal{M}_r$ near $\mathbf{B}^\ast$.
If $\widehat\Sigma = n^{-1} \widetilde{\mathbf{X}}^\top \widetilde{\mathbf{X}}$ is of full rank, the condition holds with $c_{\mathbf{X}}^2 = \sigma_{\mathrm{min}}^2(\widehat\Sigma)$, where $\sigma_{\mathrm{min}}^2(\widehat\Sigma)$ denotes the smallest eigenvalue of $\widehat\Sigma$.

\subsection{Proofs of Theorems} \label{sec:proof-consistency}

To utilize the low-dimensional manifold structure of $\mathcal{M}$ near $\mathbf{B}^\ast$, we introduce some notations:
Let $\mathbf{B}^\ast = \mathbf{U}^\ast \mathbf{S}^\ast \mathbf{V}^{\ast \top}$ be the singular value decomposition of $\mathbf{B}^\ast$ with orthonormal matrices $\mathbf{U}^\ast \in \mathbb{R}^{p \times r}$ and $\mathbf{V}^\ast \in \mathbb{R}^{M \times r}$ and the diagonal matrix $\mathbf{S}^\ast = \mathrm{diag}(\sigma_1(\mathbf{B}^\ast), \ldots, \sigma_{r}(\mathbf{B}^\ast))$ satisfying $\sigma_1(\mathbf{B}^\ast) \geq \cdots \geq \sigma_{r}(\mathbf{B}^\ast) > 0$. 
Denote the tangent space of $\mathcal{M}_r$ near $\mathbf{B}^\ast$ by $\mathcal{T}_{\ast} = \mathcal{T}_{\mathbf{B}^\ast}(\mathcal{M}_r)$ and the projection of $\mathbf{A} \in \mathbb{R}^{p \times M}$ onto the tangent space $\mathcal{T}_{\ast}$ by
\begin{align} \label{pf:thm-proj-tangent}
    \begin{split}
        \mathcal{P}_\ast (\mathbf{A}) 
        &= \Pi_{\mathbf{U}^\ast} \mathbf{A} + \mathbf{A} \Pi_{\mathbf{V}^\ast} - \Pi_{\mathbf{U}^\ast} \mathbf{A} \Pi_{\mathbf{V}^\ast} \\
        &= \Pi_{\mathbf{U}^\ast} \mathbf{A} \Pi_{\mathbf{V}^\ast}^\perp + \Pi_{\mathbf{U}^\ast}^\perp \mathbf{A} \Pi_{\mathbf{V}^\ast} + \Pi_{\mathbf{U}^\ast} \mathbf{A} \Pi_{\mathbf{V}^\ast},
    \end{split}
\end{align}
where $\Pi_{\mathbf{U}^\ast} = \mathbf{U}^\ast \mathbf{U}^{\ast \top}$ and $\Pi_{\mathbf{V}^\ast} = \mathbf{V}^\ast \mathbf{V}^{\ast \top}$ \citep{koch2007dynamical, absil2015low}. 
Then, the orthogonal complement of $\mathbf{A}$ in the normal space at $\mathbf{B}^\ast$ is given by $\mathcal{P}_\ast^\perp (\mathbf{A}) = \Pi_{\mathbf{U}^\ast}^\perp \mathbf{A} \Pi_{\mathbf{V}^\ast}^\perp$, 
where $\Pi_{\mathbf{U}^\ast}^\perp = I - \Pi_{\mathbf{U}^\ast}$ and $\Pi_{\mathbf{V}^\ast}^\perp = I - \Pi_{\mathbf{V}^\ast}$ with the identity matrix $I$.

\begin{lem} \label{lem:large-deviation}
    Suppose Assumptions \ref{assump:support} and \ref{assump:cov-kernel}, along with the above Conditions \ref{condition:regularization} and \ref{condition:RSC}, hold.
    Let $\widetilde{\mathbf{X}} \in \mathbb{R}^{n \times p}$ be the design matrix having the $(i, j)$-element with $X_{i, j} - \bar{X}_j$ and $\mathcal{E} \in \mathbb{R}^{n \times M}$ is the pseudo-error matrix having $(i, m)$-element with $\epsilon_{i, m} = Q_{Y_i}(u_m) - Q(u_m | X_i)$ as defined in \eqref{OT:pseudo-error}. 
    Then, we have $\| \Pi_{\mathbf{U}^\ast} (\widetilde{\mathbf{X}}^\top \mathcal{E}) \Pi_{\mathbf{V}^\ast} \|_F = O_P(\sqrt{nM})$, $\| \Pi_{\mathbf{U}^\ast} (\widetilde{\mathbf{X}}^\top \mathcal{E}) \Pi_{\mathbf{V}^\ast}^\perp \|_F = O_P(\sqrt{nr(M-r)})$, $\| \Pi_{\mathbf{U}^\ast}^\perp (\widetilde{\mathbf{X}}^\top \mathcal{E}) \Pi_{\mathbf{V}^\ast} \|_F = O_P(\sqrt{n(p-r)r})$, and $\| \Pi_{\mathbf{U}^\ast}^\perp (\widetilde{\mathbf{X}}^\top \mathcal{E}) \Pi_{\mathbf{V}^\ast}^\perp \|_F, = O_P(\sqrt{n(p-r)(M-r)})$. 
\end{lem}

\begin{proof}[Proof of Lemma \ref{lem:large-deviation}]
    We note that
    \begin{align}
        \begin{split}
            \| \Pi_{\mathbf{U}^\ast} (\widetilde{\mathbf{X}}^\top \mathcal{E}) \Pi_{\mathbf{V}^\ast} \|_F^2
            &= \left\| \Pi_{\mathbf{U}^\ast} \Bigg( \sum_{i = 1}^n \widetilde{X}_i \epsilon_i^\top \Bigg) \Pi_{\mathbf{V}^\ast} \right\|_F^2\\
            &= \left\| \mathbf{U}^\ast \Bigg( \sum_{i = 1}^n \big( {\mathbf{U}^\ast}^\top \widetilde{X}_i \big) \otimes (\mathbf{V}^\ast \varepsilon_i) \Bigg) {\mathbf{V}^\ast}^\top \right\|_F^2\\
            &= \sum_{j = 1}^r \sum_{k = 1}^r \Bigg( \sum_{i = 1}^n {U_j^\ast}^\top\widetilde{X}_i \epsilon_i^\top V_k^\ast \Bigg)^2,
        \end{split}
    \end{align}
    where $U_j^\ast$ and $V_k^\ast$ are the $j$-th and $k$-th column vectors of $\mathbf{U}^\ast$ and $\mathbf{V}^\ast$, respectively. 
    To get the above last equality, we used the fact that $\| \mathbf{U}^\ast \mathbf{A} {\mathbf{V}^\ast}^\top \|_F^2 = \mathrm{tr}\big( \mathbf{V}^\ast \mathbf{A}^\top {\mathbf{U}^\ast}^\top \mathbf{U}^\ast \mathbf{A} {\mathbf{V}^\ast}^\top\big)  = \mathrm{tr}\big( \mathbf{A}^\top \mathbf{A} {\mathbf{V}^\ast}^\top \mathbf{V}^\ast \big) = \mathrm{tr}(A^\top A)$ holds for any $\mathbf{A} \in \mathbb{R}^{r \times r}$, where ${\mathbf{U}^\ast}^\top \mathbf{U}^\ast$ and ${\mathbf{V}^\ast}^\top \mathbf{V}^\ast$ are identity matrices in $\mathbb{R}^{r \times r}$.
    
    Let $\| W \|_{\psi_2} = \inf\{ k \geq 0: \mathit{E} e^{ (W - \mathit{E}W)^2 / k^2} \leq 2 \}\}$ denote the sub-Gaussian norm of a random variable $W$.
    By Assumption \ref{assump:support}, it can be verified that a random variable $W_{j, k}^\ast = {U_j^\ast}^\top X \epsilon^\top V_k^\ast$ is sub-Gaussian with mean zero with $\| W_{j, k} \|_{\psi_2} \leq \sqrt{M} \tau_X \Delta_{\mathcal{S}} $ for each pair of $(j, k)$.
    Then, it follows from Hoeffding's inequality that
    \begin{align}
        \begin{split}
            P\left( \Bigg| \sum_{i = 1}^n \sum_{i = 1}^n {U_j^\ast}^\top X_i \epsilon_i^\top V_k^\ast \Bigg| \geq t \right) \leq 2 \exp\bigg( -\frac{\kappa t^2}{n M \tau_X^2 \Delta_{\mathcal{S}}^2 }\bigg)
        \end{split}
    \end{align}
    holds for every $t > 0$ with some absolute constant $\kappa > 0$, i.e., $n^{-1} \sum_{i = 1}^n {U_j^\ast}^\top\widetilde{X}_i \epsilon_i^\top V_k^\ast = O_P(\sqrt{M/n})$.
    This indicates that $\| \Pi_{\mathbf{U}^\ast} (\widetilde{\mathbf{X}}^\top \mathcal{E}) \Pi_{\mathbf{V}^\ast} \|_F = O_P(\sqrt{nM})$. 
    The rest of them can also be verified similarly.
\end{proof}

\subsubsection{Proof of Theorem \ref{thm:consistency}} \label{sec:proof-rate}
    Writing $L_n(\mathbf{B}) = \sum_{i = 1}^n M^{-1} \sum_{m = 1}^M \big( Q_{Y_i}(u_m) - \overline{Q}_Y(u_m) - \sum_{j = 1}^p  b_{j, m} (X_{i, j} - \bar{X}_j) \big)^2$, we have 
    \begin{align} \label{pf:thm-objectve-function-diff}
        L_n(\mathbf{B}) - L_n(\mathbf{B}^\ast) = M^{-1} \| \widetilde{\mathbf{X}}(\mathbf{B} - \mathbf{B}^\ast) \|_F^2 + 2 M^{-1} \langle \mathbf{B} - \mathbf{B}^\ast, \widetilde{\mathbf{X}}^\top \mathcal{E} \rangle_F,
    \end{align}
    where $\langle \mathbf{A}, \mathbf{B} \rangle_F = \mathrm{tr} (\mathbf{A}^\top \mathbf{B})$ denotes the Frobenius inner product of two matrices. 
    Recalling the definition of $\widehat{\mathbf{B}} \in \mathcal{M}_r$ in \eqref{optim:discrete-penalized}, we have
    \begin{align} \label{pf:thm-objectve-function-penalty}
        L_n(\widehat{\mathbf{B}}) + \lambda \mathcal{J}(\widehat{\mathbf{B}}) \leq L_n(\mathbf{B}^\ast) + \lambda \mathcal{J}(\mathbf{B}^\ast).
    \end{align}   
    With the orthogonal decomposition in \eqref{pf:thm-proj-tangent}, we also have
    \begin{align} \label{pf:thm-proj}
        \langle \widehat{\mathbf{B}} - \mathbf{B}^\ast, \widetilde{\mathbf{X}}^\top \mathcal{E} \rangle_F = \langle \widehat{\mathbf{B}} - \mathbf{B}^\ast, \mathcal{P}_\ast(\widetilde{\mathbf{X}}^\top \mathcal{E}) \rangle_F + \langle \widehat{\mathbf{B}} - \mathbf{B}^\ast, \mathcal{P}_\ast^\perp(\widetilde{\mathbf{X}}^\top \mathcal{E}) \rangle_F.
    \end{align}

    Suppose $M^{-1/2} \| \widehat{\mathbf{B}} - \mathbf{B}^\ast \|_F \geq \delta = \frac{\sigma_{r}(\mathbf{B}^\ast)}{2} M^{-1/2}$.
    We note that Lemma 6 of \citep{huang2024framework} implies $\big\| \mathcal{P}_\ast^\perp (\widehat{\mathbf{B}} - \mathbf{B}^\ast) \big\|_F \leq \frac{2}{\sigma_{r}(\mathbf{B}^\ast)} \| \mathcal{P}_\ast (\widehat{\mathbf{B}} - \mathbf{B}^\ast) \|_F^2$. 
    From this together with \eqref{pf:thm-proj}, the Cauchy–Schwarz inequality gives
    \begin{align} \label{pf:thm-cross-term}
    \begin{split}
        \big| \langle \widehat{\mathbf{B}} - \mathbf{B}^\ast, \widetilde{\mathbf{X}}^\top \mathcal{E} \rangle_F \big|
        &\leq \| \mathcal{P}_\ast(\widetilde{\mathbf{X}}^\top \mathcal{E}) \|_F \cdot \| \widehat{\mathbf{B}} - \mathbf{B}^\ast \|_F
        + \frac{2}{\sigma_{r}(\mathbf{B}^\ast)} \| \mathcal{P}_\ast^\perp(\widetilde{\mathbf{X}}^\top \mathcal{E}) \|_F \cdot \| \widehat{\mathbf{B}} - \mathbf{B}^\ast \|_F^2.
    \end{split}
    \end{align}
    Moreover, by substituting $\widehat{\mathbf{B}}$ for $\mathbf{B}$ in \eqref{pf:thm-objectve-function-diff} and combining it with \eqref{pf:thm-objectve-function-penalty}, it also follows from \eqref{pf:thm-proj} and Conditions \ref{condition:regularization}--\ref{condition:RSC} with any $\epsilon \in (0, \delta)$ that
    \begin{align} \label{pf:thm-objective-function}
        \begin{split}
            n c_{\mathbf{X}}^2 M^{-1} \| \mathbf{B} - \mathbf{B}^\ast \|_F^2
            &\leq M^{-1} \| \widetilde{\mathbf{X}}(\widehat{\mathbf{B}} - \mathbf{B}^\ast) \|_F^2 \\
            &\leq 2 M^{-1} \big| \langle \widehat{\mathbf{B}} - \mathbf{B}^\ast, \widetilde{\mathbf{X}}^\top \mathcal{E} \rangle_F \big| + \lambda C_{\mathcal{J}, p} M^{-1/2} \| \widehat{\mathbf{B}} - \mathbf{B}^\ast \|_F.
            \end{split}
    \end{align}
    Therefore, from \eqref{pf:thm-cross-term} and \eqref{pf:thm-objective-function}, we conclude that
    \begin{align} \label{pf:thm-stochastic-upper-bound}
        \begin{split}
            M^{-1/2}\| \widehat{\mathbf{B}} - \mathbf{B}^\ast \|_F 
            &\leq \max \left\{ \delta, \frac{n^{-1} \big( 2 M^{-1/2}  \| \mathcal{P}_\ast(\widetilde{\mathbf{X}}^\top \mathcal{E}) \|_F + \lambda C_{\mathcal{J}, p} \big)}{c_{\mathbf{X}}^2 - \frac{4}{\sigma_r(\mathbf{B}^\ast)} n^{-1} M^{-1/2} \| \mathcal{P}_\ast^\perp(\widetilde{\mathbf{X}}^\top \mathcal{E}) \|_F} \right\}.
        \end{split}
    \end{align}

    To get the rate of convergence, we determine the magnitude of the second component within the maximum sign in the right side of \eqref{pf:thm-stochastic-upper-bound}.
    To this end, we note that
    \begin{align} \nonumber
        \begin{split}
            \| \mathcal{P}_\ast(\widetilde{\mathbf{X}}^\top \mathcal{E}) \|_F 
            &= \| \Pi_{\mathbf{U}^\ast} (\widetilde{\mathbf{X}}^\top \mathcal{E}) \Pi_{\mathbf{V}^\ast} \|_F + \| \Pi_{\mathbf{U}^\ast} (\widetilde{\mathbf{X}}^\top \mathcal{E}) \Pi_{\mathbf{V}^\ast}^\perp \|_F + \| \Pi_{\mathbf{U}^\ast}^\perp (\widetilde{\mathbf{X}}^\top \mathcal{E}) \Pi_{\mathbf{V}^\ast} \|_F,\\
            \| \mathcal{P}_\ast^\perp(\widetilde{\mathbf{X}}^\top \mathcal{E}) \|_F 
            &= \| \Pi_{\mathbf{U}^\ast}^\perp (\widetilde{\mathbf{X}}^\top \mathcal{E}) \Pi_{\mathbf{V}^\ast}^\perp \|_F.
        \end{split}
    \end{align}
    Finally, applying Lemma \ref{lem:large-deviation}, we conclude that $\| \widehat{\mathbf{B}} - \mathbf{B}^\ast \|_F = O_P\big(\sqrt{r (p + M) / n} + \sqrt{\lambda p M / n} \big)$.

\subsubsection{Proof of Theorem \ref{cor:integral}} \label{sec:proof-integral}

By the mean value theorem, there exists $\xi_m \in (u_{m-1}, u_m)$ such that
\begin{align*}
    \int_{u_{m - 1}}^{u_m} \|\hat\beta^\ast(u) - \beta^\ast(u)\|_2^2 \, \mathrm{d}u
    = \frac{1}{M} \|\hat\beta^\ast(\xi_m) - \beta^\ast(\xi_m)\|_2^2
\end{align*}
for each $m = 1, \ldots, M$.
It follows that
\begin{align} \label{pf-cor-integral:mvt}
    \begin{split}
        \int_0^1 \|\hat\beta^\ast(u) - \beta^\ast(u)\|_2^2 \, \mathrm{d}u
        &= \frac{1}{M} \sum_{m=1}^{M} \|\hat\beta^\ast(\xi_m) - \beta^\ast(\xi_m)\|_2^2
    \end{split}
\end{align}
Since $\hat\beta^\ast(\xi_m)$ is the linear interpolation of $\hat\beta(u_{m - 1})$ and $\hat\beta(u_m)$, we note that $\|\hat\beta^\ast(\xi_m) - \hat\beta(u_m)\|_2 \leq \|\hat\beta(u_m) - \hat\beta(u_{m-1})\|_2 \leq \|\hat\beta(u_m) - \beta^\ast(u_m)\|_2 + \|\beta^\ast(u_m) - \beta^\ast(u_{m-1})\|_2 + \|\beta^\ast(u_{m-1}) - \hat\beta(u_{m-1})\|_2$.
Moreover, $\|\hat\beta(u_m) - \beta^\ast(\xi_m)\|_2 \leq \| \hat\beta(u_m) - \beta^\ast(u_m) \|_2 + \| \beta^\ast(u_m) - \beta^\ast(\xi_m) \|_2$. 
It follows from Jensen's inequality that 
\begin{align} \label{pf-cor-integral:approx}
    \begin{split}
        \frac{1}{M} \sum_{m=1}^{M} \|\hat\beta^\ast(\xi_m) - \beta^\ast(\xi_m)\|_2
        &\leq 3\Bigg(\frac{1}{M} \sum_{m=1}^{M} \|\hat\beta(u_m) - \beta^\ast(u_m)\|_2^2 \Bigg)^{1/2} + \frac{2}{M} \mathrm{TV}(\beta^\ast) \\
        &= O( M^{-1/2}) \cdot \| \widehat{\mathbf{B}} - \mathbf{B}^\ast \|_F
            + O(M^{-1})
    \end{split}
\end{align}
Since $\sum_{m=1}^{M} \|\hat\beta^\ast(\xi_m) - \beta^\ast(\xi_m)\|_2^2 \leq \big(\sum_{m=1}^{M} \|\hat\beta^\ast(\xi_m) - \beta^\ast(\xi_m)\|_2 \big)^2$, combining \eqref{pf-cor-integral:mvt} and \eqref{pf-cor-integral:approx} with Theorem \ref{thm:consistency}, we obtain the desired result.

\subsubsection{Proof of Theorem \ref{cor:mspe}} \label{sec:proof-mspe}

It follows from \eqref{model:quantile-reparam} that
\begin{align}
    \begin{split}    
        \big(Q_\textrm{new}(u) - \widehat{Q}_\textrm{new}(u) \big)^2
        &\leq 5 \bigg( \big( \overline{Q}_Y(u) -  \mathrm{E} \, Q_Y(u)\big)^2 + \Big| \big(\hat\beta^\ast(u) - \beta^\ast(u) \big)^\top (X_{\textrm{new}} - \mu) \Big|^2\\
        &\quad\qquad + \Big| \big(\hat\beta^\ast(u) - \beta^\ast(u) \big)^\top (\bar{X} - \mu) \Big|^2 + \Big| \beta^\ast(u)^\top (\bar{X} - \mu) \Big|^2 + \varepsilon(u)^2 \bigg)
    \end{split} \nonumber
\end{align}
holds for all $u \in [0, 1]$. 
We first note that 
\begin{align}
    \begin{split}
        \mathrm{E} \bigg( \int_0^1 \big( \overline{Q}_Y(u) -  \mathrm{E} \, Q_Y(u)\big)^2 \, \mathrm{d}u \bigg)
        &= \int_0^1 \mathrm{E}\big( \overline{Q}_Y(u) -  \mathrm{E} \, Q_Y(u)\big)^2 \, \mathrm{d}u\\
        &\leq \frac{1}{n} \int_0^1 \mathrm{E}\big( Q_Y(u)^2\big) \, \mathrm{d}u\\
        &= \frac{1}{n} \mathrm{E} \bigg( \int_\mathcal{S} s^2 \, \mathrm{d}Y(s) \bigg)
        = O(n^{-1}),
    \end{split} \nonumber
\end{align}
where we applied Fubini's theorem to the first and second equalities. 
It follows from Markov's inequality that $\int_0^1 \big( \overline{Q}_Y(u) -  \mathrm{E} \, Q_Y(u)\big)^2 \, \mathrm{d}u = O_P(n^{-1})$. 
We also get from Fubini's theorem that $\mathrm{E} \big( \int_0^1 \varepsilon(u)^2 \, \mathrm{d}u \big) = \sigma_\varepsilon^2$ and that
\begin{align}
    \begin{split}
        \mathrm{E} \bigg( \int_0^1 \Big| \big(\hat\beta^\ast(u) - \beta^\ast(u) \big)^\top (X_{\textrm{new}} - \mu) \Big|^2 \, \mathrm{d}u \,\bigg|\, \mathcal{X}_n \bigg) 
        &\leq \int_0^1 \big(\hat\beta^\ast(u) - \beta^\ast(u) \big)^\top \Sigma \big(\hat\beta^\ast(u) - \beta^\ast(u) \big) \, \mathrm{d}u\\
        &\leq \sigma_{\mathrm{max}}^2(\Sigma) \cdot \| \hat\beta^\ast - \beta^\ast \|_2^2.
    \end{split} \nonumber
\end{align}
where $\sigma_{\mathrm{max}}^2(\Sigma)$ denotes the largest eigenvalue of $\Sigma = \mathrm{Var}(X)$. 
Moreover, the Cauchy–Schwarz inequality gives
\begin{align}
    \begin{split}
        \int_0^1 \Big| \big(\hat\beta^\ast(u) - \beta^\ast(u) \big)^\top (\bar{X} - \mu) \Big|^2 \, \mathrm{d}u
        &\leq \| \bar{X} - \mu \|^2 \cdot \| \hat\beta^\ast - \beta^\ast \|_2^2 \\
        &= O_P(p / n) \cdot \| \hat\beta^\ast - \beta^\ast \|_2^2.
    \end{split} \nonumber
\end{align}
To get the above probabilistic bound, we used Markov's inequality with $\mathrm{E} \| \bar{X} - \mu \|^2 \leq \mathrm{tr}(\Sigma) / n$. 
Similarly, $\int_0^1 \big| \beta^\ast(u)^\top (\bar{X} - \mu) \big|^2 \, \mathrm{d}u = O_P(p / n)$. 
Combining all the above results with Jensen's inequality, we conclude that
\begin{align}
    \begin{split}
        \mathit{E} \big[ d_{W_2}(Y_\textrm{new}, \widehat{Y}_\textrm{new}) \,|\, \mathcal{X}_n \big] 
        &\leq \Big( \mathit{E} \big[ d_{W_2}(Y_\textrm{new}, \widehat{Y}_\textrm{new})^2 \,|\, \mathcal{X}_n \big] \Big)^{1/2} \\
        &= O(\sigma_\varepsilon) + O_P\big( \sqrt{r (p + M) / n} + \sqrt{\lambda  p M / n} \big) + O(M^{-1/2}).
    \end{split} \nonumber
\end{align}

\bibliographystyle{abbrv}
\bibliography{ref-bib}

\end{document}